\documentclass{lmcs}

\usepackage{amsmath,proof,amssymb}
\usepackage{tikz,chngcntr}
\usetikzlibrary{automata,positioning}

\newcommand{\defor}{\; | \;}
\newcommand{\defeq}{\; ::= \;}

\usepackage{stmacros}
\newcommand{\sillSLeq}{SILL_{S{\leq}}}
\newcommand{\sillS}{SILL_{S}}

\theoremstyle{definition}
\newtheorem{case}{Case}
\newtheorem{subcase}{Subcase}
\counterwithin*{case}{thm}
\counterwithin*{subcase}{case}

\begin{document}

\title{Manifestly Phased Communication via Shared Session Types}
\author{Chuta Sano}
\address{Department of Computer Science, Carnegie Mellon University, Pittsburgh, USA}
\email{chutasano@gmail.com}

\author{Stephanie Balzer}
\address{Department of Computer Science, Carnegie Mellon University, Pittsburgh, USA}
\email{balzers@cs.cmu.edu}

\author{Frank Pfenning}
\address{Department of Computer Science, Carnegie Mellon University, Pittsburgh, USA}
\email{fp@cs.cmu.edu}

\begin{abstract}
  Session types denote message protocols between concurrent processes, allowing a type-safe expression of inter-process communication.
  Although previous work demonstrate a well-defined notion of subtyping where processes have different perceptions of the protocol, these
  formulations were limited to linear session types where each channel of communication has a unique provider and client. In this paper, we
  extend subtyping to shared session types where channels can now have multiple clients instead of a single client. We demonstrate that this
  generalization can statically capture protocol requirements that span multiple phases of interactions of a client with a shared service
  provider, something not possible in prior proposals.  Moreover, the phases are manifest in the type of the client.
\end{abstract}

\maketitle{}

\footnotetext[1]{This is a revised and extended version of a paper presented at COORDINATION 2021~\cite{Sano21coord}. The main changes in this
version include an additional example demonstrating phasing in~\autoref{sec:dd-detection} and a formalization of a system implementing our
work in~\autoref{sec:metatheory} with proofs of relevant metatheorems in the Appendix.}

\section{Introduction}
\label{sec:intro}
Session types prescribe bidirectional communication protocols between concurrent processes~\cite{Honda93concur,Honda98esop}. Variations of this type system were later given logical correspondences with \emph{intuitionistic}~\cite{Caires10concur} and \emph{classical}~\cite{Wadler12icfp} linear logic where proofs correspond to programs and cut reduction to communication. This correspondence mainly provides an interpretation of \emph{linear session types}, which denote sessions with exactly one client and one provider. \emph{Shared session types}, which encode communication between multiple clients and one provider, were proposed with a \emph{sharing semantics} interpretation in prior work \cite{Balzer17icfp}. Clients communicating along a shared channel follow an \emph{acquire-release} discipline where they must first \emph{acquire} exclusive access to the provider, communicate linearly, and then finally \emph{release} the exclusive access, allowing other clients to acquire.

However, not all protocols that follow this acquire-release paradigm are safe; if a client that successfully acquires some shared channel of
type $A$ releases it at an unrelated type $B$, other clients that are blocked while trying to acquire will still see the channel as type $A$
while the provider will see the channel as type $B$. To resolve this, we require an additional constraint that clients must release at the
same type at which it acquired. This is formally expressed in~\cite{Balzer17icfp} as the \emph{equi-synchronizing} constraint, which
statically verifies that session types encode communication which does not release at a different type than its original. Although shared
session types serve an important role in making session typed process calculi theory applicable to practical scenarios, they cannot express
\emph{phases}, or protocols across successive acquire-release cycles, due to the equi-synchronizing constraint being too restrictive (see
\autoref{sec:phasing})~\cite{Sano19ms}.

We demonstrate that subtyping, first formalized in the session-typed process calculi setting by Gay and Hole~\cite{Gay05acta}, and its behavior across the two linear and shared modalities provide the groundwork for an elegant relaxation of the equi-synchronizing constraint, allowing for phases to be \emph{manifest} in the session type. In message passing concurrency, subtyping allows a client and provider to safely maintain their own local views on the session type (or protocol) associated with a particular channel. Although previous work~\cite{Gay05acta,Acay16itrs} investigate subtyping in the purely linear session type setting, we found that extending these results to the linear and shared session type setting as in~\cite{Balzer17icfp} yields very powerful results with both practical and theoretical significance. 

In this paper, we propose $\sillSLeq$, an extension of $\sillS$~\cite{Balzer17icfp} with subtyping, and show that metatheorems such as progress and preservation that hold true in $\sillS$ still hold true in $\sillSLeq$. We in particular introduce the \emph{subsynchronizing} constraint, a relaxation of the equi-synchronizing constraint, which denote under what conditions clients and providers can safely disagree on the protocol in shared communnication.

\noindent The main contributions of this paper include:
\begin{itemize}
  \item A full formalization of a subtyping relation for shared session types and their metatheory.
  \item The introduction of the subsynchronizing constraint, a relaxation of the equi-synchronizing constraint.
  \item Demonstration of $\sillSLeq$, a message passing concurrency system with shared subtyping, along with proofs of the progress and preservation theorems.
  \item Illustrations of practical examples in this richer type system, further bridging the gap between session-typed process calculi and practical programming languages.
\end{itemize}

\noindent The rest of the paper proceeds as follows: \autoref{sec:background} provides a brief introduction to linear and shared
session-typed message-passing concurrency. \autoref{sec:esync-phasing} demonstrates the inability of prior systems to express phasing and
motivates our approach. \autoref{sec:subtyping} provides an introduction to linear subtyping along with an attempt to extend the relation to
the shared setting. \autoref{sec:phasing} introduces the notion of phasing and the subsynchronizing judgment. \autoref{sec:metatheory}
presents a message passing concurrent system using our type system and the corresponding progress and preservation statements.
\autoref{sec:related} discusses related work. \autoref{sec:conclusion} concludes the paper with some points of discussion and future work.
Finally, the Appendix contains detailed proofs of metatheorems and lemmas that we introduce in the paper.

\section{Background}
\label{sec:background}

\subsection{Linear Session Types}
\label{sec:background-linear}
Based on the correspondence established between intuitionistic linear logic and the session-typed $\pi$-calculus~\cite{Caires10concur,Toninho15phd} we can interpret a intuitionistic \emph{linear} sequent
$$A_1, A_2, \ldots, A_n \vdash B$$
as the typing judgment for a process $P$ by annotating the linear propositions with channel names:
$$\underbrace{a_1:A_1, a_2:A_2, \ldots, a_n:A_n}_\Delta \vdash P :: (b:B)$$
Interpreted as a typing judgment, we say that process $P$ \emph{provides} a session of type $B$ along channel $b$ while \emph{using} channels $a_1, \ldots, a_n$ with session types $A_1, \ldots, A_n,$ respectively. Interpreted as a sequent, we say that $P$ is a proof of some proposition $B$ with hypotheses $A_1, \ldots, A_n$. Following linear logic, the context $\Delta$ is restricted and rejects contraction and weakening. Programatically, this means that linear channels cannot be aliased nor freely deleted -- they must be fully consumed exactly once.

Since the session type associated with a channel denotes a bidirectional protocol, each connective has two operational interpretations -- one from the perspective of the provider and one from the client. This operationally dual interpretation results in a schema where for any connective, either the client or provider will send while the other will receive as summarized in Table~\ref{tab:linear-types}.

For example, a channel of type $A \otimes 1$ requires that the provider sends a channel of type $A$ and proceeds as type $1$ while the client receives a channel of type $A$ and proceeds as $1$. The multiplicative unit $1$ denotes the end of the protocol -- the provider must terminate and close its channel while a client must wait for the channel to be closed. A channel of type $\intch{\overline{l:A}}$ ($n$-nary internal choice) requires the provider to choose and send a label $i$ in $\overline{l}$ and proceed as $A_i$ while the client must receive and branch on some label $i$ and proceed as $A_i$. Similarly, a channel of type $\extch{\overline{l:A}}$ requires the client to choose and send a label and the provider to receive and branch on a label. The \emph{continuation type} of some session type refers to the type after a message exchange; for example, $B$ would be the continuation type of $A \otimes B$ and similarly $A_i$ of $\intch{\overline{l:A}}$ for some $i$ in $\overline{l}$. The unit $1$ does not have a continuation type since it marks the end of communication.
{
\setlength{\tabcolsep}{5pt}
\begin{table}[ht]
	\begin{tabular}{@{}llll@{}}
		Type & Interpretation from provider & Interpretation from client & Continuation\\
		\hline
		$1$ & Close channel (terminate) & Wait for channel to close  & - \\
		$A \otimes B$ & Send channel of type $A$ & Receive channel of type $A$  & $B$ \\
		$A \multimap B$ & Receive channel of type $A$ & Send channel of type $A$ & $B$ \\
		$\intch{\overline{l:A}}$ & Send a label $i \in \overline{l}$ & Receive and branch on $i \in \overline{l}$ & $A_i$ \\
		$\extch{\overline{l:A}}$ & Receive and branch on $i \in \overline{l}$ & Send a label $i \in \overline{l}$ & $A_i$
	\end{tabular}
	\centering
	\caption{A summary of the linear connectives and their operational interpretations}
	\label{tab:linear-types}
\end{table}
}

We consider a session type denoting the interaction with a provider of a queue of integers, which we will develop throughout the paper:

\begin{small}
\begin{align*}
	\textbf{queue} = \& \{\mathit{enqueue}: &\text{int} \supset \textbf{queue}, \\
	\mathit{dequeue}: & \intch{\mathit{some}: \text{int} \land \textbf{queue}, \mathit{none}: \textbf{queue}} \}
\end{align*}
\end{small}%
where we informally adopt value input and output $\supset$ and $\land$~\cite{Toninho15phd} as value analogues to channel input and output $\multimap$ and $\otimes$, respectively, which are orthogonal to the advancements in this work. Following this protocol, a client must send a label $\mathit{enqueue}$ or $\mathit{dequeue}$. If it chooses $\mathit{enqueue}$, it must send an int and then recur, and on the other hand, if it chooses $\mathit{dequeue}$, it will receive either some int as indicated by the $\mathit{some}$ branch of the internal choice or nothing as indicated by the $\mathit{none}$ branch. In either case, we let the queue recur\footnote{We do not consider termination to more easily align with later examples.}. Dually, a server must first receive a label $\mathit{enqueue}$ or $\mathit{dequeue}$ from the client. If it receives an $\mathit{enqueue}$, it will receive an int and then recur. If it receives a $\mathit{dequeue}$ instead, it must either send a $\mathit{some}$ label followed by the appropriate int and then recur or send a $\mathit{none}$ label and then recur.

We adopt an \emph{equi-recursive}~\cite{Crary99pldi} interpretation which requires that recursive session types be \emph{contractive}~\cite{Gay05acta}, guaranteeing that there are no messages associated with the unfolding of a recursive type. This in particular requires that we reason about session types \emph{coinductively}.

We now attempt to encode a protocol representing an auction based on~\cite{Das21csf}. An auction transitions between the bidding phase where clients are allowed to place bids and the collecting phase where a winner is given the item while all the losers are refunded their respective bids.

\begin{small}
\begin{align*}
	\textbf{bidding} =
	  \&\{\mathit{bid}: &\oplus\{\mathit{ok}: \text{id} \supset \text{money} \supset \textbf{bidding}, \\
  	&\mathit{collecting}: \textbf{collecting} \}\} \\
 	\textbf{collecting} =
  	\&\{ \mathit{collect}: \text{id}\supset &\oplus\{\mathit{prize}: \text{item} \land \textbf{bidding}, \\
  	&\quad \mathit{refund}: \text{money} \land \textbf{bidding}, \\
  	&\quad \mathit{bidding}: \textbf{bidding}\}\}
\end{align*}
\label{ex:lin-auction}
\end{small}%
In this example, we make the bidding phase and collecting phase explicit by separating the protocol into \textbf{bidding} and \textbf{collecting}. Beginning with \textbf{bidding}, a client must send a $\mathit{bid}$ label~\footnote{The currently unnecessary unary choice will be useful later.}. The provider will either respond with an $\mathit{ok}$, allowing the client to make a bid by sending its id, money, and then recursing back to \textbf{bidding}, or a $\mathit{collecting}$, indicating that the auction is in the collecting phase and thereby making the client transition to \textbf{collecting}.

For \textbf{collecting}, the client must send a $\mathit{collect}$ label. For ease of presentation, we require the client to also send its id immediately, giving enough information to the provider to know if the client should receive a $\mathit{prize}$ or a $\mathit{refund},$ along with $\mathit{bidding}$ if the client is in the wrong phase. The $\mathit{prize}$ branch covers the case where the client won the previous bid, the $\mathit{refund}$ branch covers the case where the client lost the bid, and the $\mathit{bidding}$ branch informs the client that the auction is currently in the bidding phase.

Because linear channels have exactly one provider and one client, what we have described so far only encodes a single participant auction. One can assert that the provider is actually a broker to an auction of multiple participants, but that does not solve the fundamental problem, that is, encoding shared communication with multiple clients.

\subsection{Shared Session Types}
\label{sec:background-shared}
Although linear session types and their corresponding process calculi give a system with strong guarantees such as \emph{session fidelity} (preservation) and \emph{deadlock freedom} (progress), as we show in the previous section while attemping to encode an auction, they are not expressive enough to model systems with shared resources. Since multiple clients cannot simultaneously communicate with a single provider in an unrestricted manner, we adopt an \emph{acquire-release} paradigm. The only action a client can perform on a shared channel is to send an acquire request, which the provider must accept. After successfully acquiring, the client is guaranteed to have exclusive access to the provider and therefore can communicate linearly until the client releases its exclusive access.

Instead of treating the acquire and release operations as mere operational primitives, prior work~\cite{Balzer17icfp} extends the type system such that the acquire and release points are manifest in the type by stratifying session types into shared and linear types. Unlike linear channels, shared channels are unrestricted in that they can be freely aliased or deleted. In the remaining sections, we will make the distinction between linear and shared explicit by marking channel names and session type meta-variables with subscripts $L$ and $S$ respectively where appropriate. For example, a linear channel is marked $\Lo{a}$, while a shared channel is marked $\So{b}$.

Since shared channels represent unrestricted channels that must first be acquired, they are constructed by the modal upshift operator $\upls{A}$ for some $\Lo{A}$ requires clients to acquire and then proceed linearly as prescribed by $\Lo{A}$. Similarly, the modal downshift operator $\downsl{B}$ for some $\So{B}$ requires clients to release and proceed as a shared type. Type theoretically, these modal shifts mark transitions between shared to linear and vice versa. In summary, we have:

\begin{small}
\begin{align*}
	&\text{(Shared Layer)} &\So{A} \defeq &\upls{A} \\
	&\text{(Linear Layer)} &\Lo{A}, \Lo{B} \defeq &\downsl{A} \defor 1 \defor \tensor{A}{B} \defor \loli{A}{B} \defor \extch{\caselist{l}{A}} \defor \intch{\caselist{l}{A}}
\end{align*}
\end{small}%
where we emphasize that the previously defined (linear) type operators such as $\otimes$ remain only at the linear layer -- a shared session type can only be constructed by a modal upshift $\Lupls$ of some linear session type $\Lo{A}$.

As initially introduced, clients of shared channels follow an \emph{acquire-release} pattern -- they must first acquire exclusive access to the channel, proceed linearly, and then finally release the exclusive access that they had, allowing other clients of the same shared channel to potentially acquire exclusive access. The middle linear section can also be viewed as a \emph{critical region} since the client is guaranteed unique access to a shared provider process. Therefore, this system naturally supports atomic operations on shared resources.

Using shared channels, we can encode a shared queue, where there can be multiple clients interacting with the same data:

\begin{small}
\begin{align*}
	\textbf{shared\_queue} = {\color{red}\Lupls} \& \{\mathit{enqueue}: &\text{int} \supset {\color{red}\Ldownsl \textbf{shared\_queue}}, \\
	\mathit{dequeue}: & \oplus \{\mathit{some}: \text{int} \land {\color{red}\Ldownsl \textbf{shared\_queue}},\\
		&\quad \; \; \; \mathit{none}: {\color{red}\Ldownsl \textbf{shared\_queue}}\} \}
\end{align*}
\end{small}%

A client of such a channel must first send an \textbf{acquire} message, being blocked until the acquisition is successful. Upon acquisition, the client must then proceed linearly as in the previously defined linear queue. The only difference is that before recursing, the client must \textbf{release} its exclusive access, allowing other blocked clients to successfully acquire.

\section{Equi-synchronizing Rules Out Phasing}
\label{sec:esync-phasing}
We can also attempt to salvage the previous iteration of encoding (multi-participant) auctions by ``wrapping'' the previous purely linear protocol between $\Lupls$ and $\Ldownsl$.

\begin{small}
\begin{align*}
	\textbf{bidding} = {\color{red}\Lupls}
  	\&\{\mathit{bid}: &\oplus\{\mathit{ok}: \text{id} \supset \text{money} \supset {\color{red}\Ldownsl \textbf{bidding}}, \\
  	&\quad \;\; \mathit{collecting}: {\color{red}\Ldownsl \textbf{collecting}} \}\} \\
	\textbf{collecting} = {\color{red}\Lupls}
  	\&\{ \mathit{collect}: \text{id} \supset &\oplus\{\mathit{prize}: \text{item} \land {\color{red}\Ldownsl \textbf{bidding}}, \\
  	&\quad \;\; \mathit{refund}: \text{money} \land {\color{red}\Ldownsl \textbf{bidding}}, \\
  	&\quad \;\; \mathit{bidding}: {\color{red}\Ldownsl \textbf{bidding}}\}\}
\end{align*}
\end{small}%
A client to \textbf{bidding} must first acquire exclusive access as indicated by $\Lupls$, proceed linearly, and then eventually release at either \textbf{bidding} (in the $\mathit{ok}$ branch) or \textbf{collecting} (in the $\mathit{collecting}$ branch). Similarly, a client to \textbf{collecting} must first acquire exclusive access, proceed linearly, and then eventually release at \textbf{bidding} since all branches lead to \textbf{bidding}.

Unfortunately, as formulated so far, this protocol is not sound. For example, consider two auction participants $P$ and $Q$ that are both in the collecting phase and blocked trying to acquire. Suppose $P$ successfully acquires, in which case it follows the protocol linearly and eventually releases at \textbf{bidding}. Then, if $Q$ successfully acquires, we have a situation where $Q$ rightfully believes that it acquired at \textbf{collecting} but since $P$ previously released at type \textbf{bidding}, the auctioneer believes that it currently accepted a connection from \textbf{bidding}. The subsequent label sent by the client, $\mathit{collect}$ is not an available option for the provider; session fidelity has been violated.

Previous work~\cite{Balzer17icfp} addresses this problem by introducing an additional requirement that if a channel was acquired at some type $\So{A}$, all possible future releases (by looking at the continuation types) must release at $\So{A}$. This is formulated as the \emph{equi-synchronizing} constraint, defined coinductively on the structure of session types. In particular, neither \textbf{bidding} nor \textbf{collecting} are equi-synchronizing because they do not always release at the same type at which it was acquired. For \textbf{bidding}, the $\mathit{collecting}$ branch causes a release at a different type, and for \textbf{collecting}, all branches lead to a release at a different type.

A solution to the auction scenario is to unify the two phases into one:

\begin{small}
\begin{align*}
	\textbf{auction} = {\color{red}\Lupls}
	\&\{\mathit{bid}: &\oplus\{\mathit{ok}: \text{id} \supset \text{money} \supset {\color{red}\Ldownsl \textbf{auction}}, \\
		&\quad \;\; \mathit{collecting}: {\color{red}\Ldownsl \textbf{auction}} \}, \\
  \mathit{collect}: \text{id} \supset &\oplus\{\mathit{prize}: \text{item} \land {\color{red}\Ldownsl \textbf{auction}}, \\
	&\quad \;\; \mathit{refund}: \text{money} \land {\color{red}\Ldownsl \textbf{auction}}, \\
	&\quad \;\; \mathit{bidding}: {\color{red}\Ldownsl \textbf{auction}}\}\}
\end{align*}
\label{ex:shrd-auction}
\end{small}%

The type \textbf{auction} is indeed equi-synchronizing because all possible release points are at \textbf{auction}.

This presentation of the auction however loses the explicit denotation of the two phases; although the previous linear single participant version of the auction protocol can make explicit the bidding and collecting phases in the session type, the equi-synchronizing requirement forces the two phases to merge into one in the case of shared session types. In general, the requirement that all release points are equivalent prevents shared session types to encode protocols across multiple acquire-release cycles since information is necessarily ``lost'' after a particular acquire-release cycle.

\section{Subtyping}
\label{sec:subtyping}
So far, there is an implicit requirement that given a particular channel, both its provider and clients agree on its protocol or type. A relaxation of this requirement in the context of linear session types has been investigated by Gay and Hole~\cite{Gay05acta}, and in this section, we present subtyping in the context of both linear session types and shared session types.

If $A \leq B$, then a provider viewing its offering channel as type $A$ can safely communicate with a client viewing the same channel as type $B$. This perspective reveals a notion of \emph{substitutability}, where a process providing a channel of type $A$ can be replaced by a process providing $A'$ such that $A' \leq A$ and dually, a client to some channel of type $B$ can be replaced by another process using the same channel as some type $B'$ such that $B \leq B'$. The following subtyping rules, interpreted coinductively, formalize the subtyping relation between session types:

\begin{small}
\[
\infer[{\leq}_1]{1 \leq 1}{}
\quad
\infer[{\leq}_\otimes]{\tensor{A}{B} \leq \tensor{A'}{B'}}{\Lo{A} \leq \Lo{A'} && \Lo{B} \leq \Lo{B'}}
\quad
\infer[{\leq}_\multimap]{\loli{A}{B} \leq \loli{A'}{B'}}{\Lo{A'} \leq \Lo{A} && \Lo{B} \leq \Lo{B'}}
\]
\[
\infer[{\leq}_\oplus]{\intch{\caselist{l}{A}} \leq \intch{\caselist{l}{A'}, \caselist{m}{B}}}
{\forall{i} \in \overline{l} \quad \Lo{{A_i}} \leq \Lo{{A'_i}}}
\quad
\infer[{\leq}_\&]{\extch{\caselist{l}{A}, \caselist{m}{B}} \leq \extch{\caselist{l}{A'}}}
{\forall{i} \in \overline{l} \quad \Lo{{A_i}} \leq \Lo{{A'_i}}}
\]
\end{small}%

One of the notable consequences of adopting subtyping is that internal and external choices allow one side to have more labels or branches. For internal choice, since the provider sends some label, there is no harm in a client to be prepared to handle additional labels that it will never receive and vice versa for external choice. Another observation is that subtyping of session types is covariant in their continuations; following this paradigm, we can immediately define subtyping for the new type connectives $\Lupls$ and $\Ldownsl$:

\begin{small}
\[
\infer[{\leq}_{\Lupls}]{\upls{A} \leq \upls{B}}{\Lo{A} \leq \Lo{B}} \quad
\infer[{\leq}_{\Ldownsl}]{\downsl{A} \leq \downsl{B}}{\So{A} \leq \So{B}}
\]
\end{small}%

\begin{rem}
	The subtyping relation $\leq$ is a partial order.
\end{rem}

A key principle governing subtyping of session types is that \emph{ignorance is
bliss}; neither the client nor the provider need to know the precise protocol
that the other party is following.

Let us revisit the shared queue example:

\begin{small}
\begin{align*}
	\textbf{shared\_queue} = \Lupls \& \{\mathit{enqueue}: &\text{int} \supset \Ldownsl \textbf{shared\_queue}, \\
	\mathit{dequeue}: & \oplus \{\mathit{some}: \text{int} \land \Ldownsl \textbf{shared\_queue},\\
	&\quad \; \; \; \mathit{none}: \Ldownsl \textbf{shared\_queue}\} \}
\end{align*}
\end{small}%

Instead of allowing all clients to freely enqueue and dequeue, suppose we only allow certain clients to enqueue and certain clients to dequeue. With subtyping, we first fix the provider's type to be \textbf{shared\_queue}. Next, we restrict writer clients by removing the $dequeue$ label and similarly restrict reader clients by removing the $enqueue$ label:

\begin{small}
\begin{align*}
	\textbf{producer} &= \Lupls \& \{\mathit{enqueue}: \text{int} \supset \Ldownsl \textbf{producer} \} \\
	\textbf{consumer} &= \Lupls \& \{\mathit{dequeue}: \intch{\mathit{some}: \text{int} \land \Ldownsl \textbf{consumer}, \mathit{none}: \Ldownsl \textbf{consumer}} \}
\end{align*}
\end{small}%
where it is indeed the case that $\textbf{shared\_queue} \leq \textbf{producer}$ and $\textbf{shared\_queue} \leq \textbf{consumer}$, justifying both the writer and reader clients' views on the type of the channel.

We will defer the detailed discussion of the subtle interactions that occur between the notion of equi-synchronizing constraint and subtyping to \autoref{sec:phasing-ssync}. For this example however, the fact that all three types \textbf{shared\_queue}, \textbf{producer}, and \textbf{consumer} are independently equi-synchronizing is a strong justification of its soundness.

\section{Phasing}
\label{sec:phasing}
One of the most common patterns when encoding data structures and protocols via session types is to begin the linear type with an external choice. When these types recur, we are met with another external choice. A notion of \emph{phasing} emerges from this pattern, where a single phase spans from the initial external choice to the recursion.

We introduced varying versions of an auction protocol, which in its linear form (Section~\ref{ex:lin-auction}) can make explicit the two distinct phases, yet in its shared form (Section~\ref{ex:shrd-auction}) cannot due to the equi-synchronizing constraint. With subtyping however, this seems to no longer be a problem; the auctioneer can view the protocol as auction whereas the clients can independently view the protocol as \textbf{bidding} or \textbf{collecting} depending on their current phase since $\textbf{auction} \leq \textbf{bidding}$ and $\textbf{auction}\leq \textbf{collecting}$.

\begin{small}
\begin{align*}
	\text{provider}&
	\begin{cases}
    \begin{aligned}
    \textbf{auction} = \Lupls
    	\&\{\mathit{bid}: &\oplus\{\mathit{ok}: \text{id} \supset \text{money} \supset \Ldownsl \textbf{auction}, \\
      	&\quad \;\; \mathit{collecting}: \Ldownsl \textbf{auction} \}, \\
    	\mathit{collect}: \text{id} \supset &\oplus\{\mathit{prize}: \text{item} \land \Ldownsl \textbf{auction}, \\
      	&\quad \;\; \mathit{refund}: \text{money} \land \Ldownsl \textbf{auction}, \\
      	&\quad \;\; \mathit{bidding}: \Ldownsl \textbf{auction}\}\}
  	\end{aligned}
	\end{cases}
	\\
	\text{clients}&
	\begin{cases}
		\begin{aligned}
    \textbf{bidding} = \Lupls
     	\&\{\mathit{bid}: &\oplus\{\mathit{ok}: \text{id} \supset \text{money} \supset \Ldownsl \textbf{bidding}, \\
      	&\quad \;\; \mathit{collecting}: \Ldownsl \textbf{collecting} \}\} \\
    \textbf{collecting} = \Lupls
    	\&\{ \mathit{collect}: \text{id} \supset  &\oplus\{\mathit{prize}: \text{item} \land \Ldownsl \textbf{bidding}, \\
      	&\quad \;\; \mathit{refund}: \text{money} \land \Ldownsl \textbf{bidding}, \\
      	&\quad \;\; \mathit{bidding}: \Ldownsl \textbf{bidding}\}\}
	  \end{aligned}
	\end{cases}
\end{align*}
\end{small}%

Unfortunately, there is a critical issue with this solution. Since shared channels can be aliased, a client in the collecting phase can alias the channel, follow the protocol, and then ignore the released type (bidding phase) -- it can then use the previously aliased channel to communicate as if in the collecting phase. In general, the strategy of encoding phases in shared communication through a shared supertype allows malicious clients to re-enter previously encountered phases since they may internally store aliases. Thus, what we require is a subtyping relation across shared and linear modes since linear channels are restricted and in particular cannot be aliased.

We first add two new linear connectives $\Lupll$ and $\Ldownll$ that, like $\Lupls$ and $\Ldownsl$, have operationally an acquire-release semantics but enforce a linear treatment of the associated channels. Prior work~\cite{Griffith15phd} has already explored such intra-layer shifts, albeit for the purpose of enforcing synchronization in an asynchronous message-passing system. Thus for example, the protocol denoted by $\upll{A}$ requires the client to ``acquire'' as in the shared case. If the provider happens to provide a linear channel $\upll{A}$, then this merely adds a synchronization point in the communication. The more interesting case is when the provider is actually providing a shared channel, some $\upls{A}$; a client should be able to view the session type as $\upll{A}$ without any trouble. We formalize this idea to the following additional subtyping relations:

\begin{small}
\[
\infer[{\leq}_{\Lupls\Lupll}]{\upls{A} \leq \upll{B}}{\Lo{A} \leq \Lo{B}}
\quad
\infer[{\leq}_{\Ldownsl\Ldownll{}}]{\downsl{A} \leq \downll{B}}{\So{A} \leq \Lo{B}}
\quad
\infer[{\leq}_{\Lupll}]{\upll{A} \leq \upll{B}}{\Lo{A} \leq \Lo{B}}
\quad
\infer[{\leq}_{\Ldownll}]{\downll{A} \leq \downll{B}}{\Lo{A} \leq \Lo{B}}
\]
\end{small}%

Using the new connectives, we can complete the auction protocol where the two phases are manifest in the session type; a client must actually view the auction protocol linearly!

\begin{small}
\begin{align*}
	\textbf{bidding} = \Lupll
	\&\{\mathit{bid}: &\oplus\{\mathit{ok}: \text{id} \supset \text{money} \supset \Ldownll \textbf{bidding}, \\
	&\quad \;\; \mathit{collecting}: \Ldownll \textbf{collecting} \}\} \\
	\textbf{collecting} = \Lupll
	\&\{ \mathit{collect}: \text{id} \supset &\oplus\{\mathit{prize}: \text{item} \land \Ldownll \textbf{bidding}, \\
	&\quad \;\; \mathit{refund}: \text{money} \land \Ldownll \textbf{bidding}, \\
	&\quad \;\; \mathit{bidding}: \Ldownll \textbf{bidding}\}\}
\end{align*}
\end{small}%
where $\textbf{auction} \leq \textbf{bidding}$ and $\textbf{auction} \leq \textbf{collecting}$. Compared to the initially presented linear auction protocol, this version inserts the purely linear shifts $\Lupll$ and $\Ldownll$ where appropriate such that the protocol is compatible with the shared auction protocol that the auctioneer provides. Therefore, the addition of $\Lupll$ and $\Ldownll$ to our system allows a natural subtyping relation between shared session types and linear session types, where they serve as a means to safely bridge between shared and linear modalities.

\subsection{Deadlock Detection}
\label{sec:dd-detection}
Another instance where phasing naturally occurs is from centralized form of Mitchell and Merritt's distributed deadlock detection algorithm \cite{Mitchell84podc}. The algorithm assumes a distributed system with shared resources and linear nodes, where the intended behavior is that the linear nodes, encoded as linear processes, acquire particular resources, encoded as shared processes, perform appropriate computations, and then release unneeded resources as in typical distributed systems. Both nodes and resources are identified by a unique identification of type pid (process id) and rid (resource id) respectively, which as in previous examples, we take as primitives. In this system, a deadlock in the usual sense is detected when there is a cycle in the dependency graph generated by the algorithm. The centralized deadlock detection algorithm consists of a shared process that acts as an monitor that all nodes report to.

The type of this global deadlock detection monitor is given as
\begin{align*}
  \textbf{dd} = \Lupls \&\{\mathit{tryacq}: &\text{pid} \supset \text{rid} \supset \Ldownsl \textbf{dd},\\
                           \mathit{didacq}: &\text{pid} \supset \text{rid} \supset \Ldownsl \textbf{dd},\\
                           \mathit{willrel}: &\text{pid} \supset \text{rid} \supset \Ldownsl \textbf{dd}\}
\end{align*}
\noindent where the intention is that clients are expected to inform the monitor before attempting to acquire a resource (\textit{tryacq}),
after successfully acquiring a resource (\textit{didacq}), and before releasing a resource (\textit{willrel}).

As discussed in a previous work \cite{Sano19ms}, there are two phases of the protocol across successive acquire-release cycles. Using subtyping, we can represent this constraint statically:

\begin{align*}
  \textbf{dd\_start} = \Lupll \&\{\mathit{tryacq}: &\text{pid} \supset \text{rid} \supset \Ldownll \textbf{dd\_acq}, \\
                                  \mathit{willrel}: &\text{pid} \supset \text{rid} \supset \Ldownll \textbf{dd\_start}\} \\
  \textbf{dd\_acq} = \Lupll \&\{\mathit{didacq}: &\text{pid} \supset \text{rid} \supset \Ldownll \textbf{dd\_start}\}
\end{align*}
\noindent where $\textbf{dd} \leq \textbf{dd\_start}$. This session type enforces that the message following \textit{tryacq} must be
\textit{didacq} and that \textit{didacq} cannot be sent without a \textit{tryacq} on the previous acquire-release cycle. It is important to note that we are not enforcing other
desirable constraints such as whether the resource id sent by the client matches in a sequence of \textit{tryacq} followed by
\textit{didacq} (it is nonsensical for a client to attempt to acquire resource $r$ and after claim that it successfully acquired a different
resource $r'$). We believe that those additional constraints can be naturally expressed by extending refinement types~\cite{Das20concur} to
be compatible with this system.

A linear node is a process that uses a channel of type \textbf{dd\_start}; since we allow subtyping across modalities, we can spawn such a node by
passing a reference to the global monitor offering a shared channel of type \textbf{dd}, which the node can safely view to be
\textbf{dd\_start} since $\textbf{dd} \leq \textbf{dd\_start}$.

\begin{rem}
  A protocol spanning multiple phases can also be interpreted as a deterministic finite autonomata (DFA) where nodes represent the phase or the state of the protocol and edges represent choice branches. The previous auction protocol can be encoded as a two state DFA as shown in Figure~\ref{fig:dfa-auction} andd the deadlock monitor protocol can similarly be encoded as shown in Figure~\ref{fig:dfa-dd}.

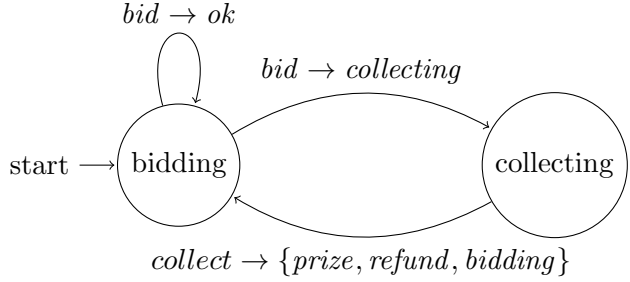
\begin{figure}
  \centering
  \begin{tikzpicture}[shorten >=1pt,node distance=5cm,on grid,auto] 
    \node[state,initial] (bid)   {bidding}; 
    \node[state] (collect) [right=of bid] {collecting}; 
    \path[->] 
      (bid) edge [loop above] node {$\mathit{bid} \rightarrow \mathit{ok}$} (bid)
      (bid) edge[bend left] node {$\mathit{bid} \rightarrow \mathit{collecting}$} (collect)
      (collect) edge[bend left] node {$collect \rightarrow \{\mathit{prize}, \mathit{refund}, \mathit{bidding}\}$} (bid);
  \end{tikzpicture}
\caption{A DFA representation of the two phases in the auction protocol, where non-branching messages are omitted for presentation purposes since they do not contribute to different protocol paths. Multiple labels enclosed in brackets as in $\{\mathit{prize}, \mathit{refund}, \mathit{bidding}\}$ mean that any of those labels can be selected.} \label{fig:dfa-auction}
\end{figure}

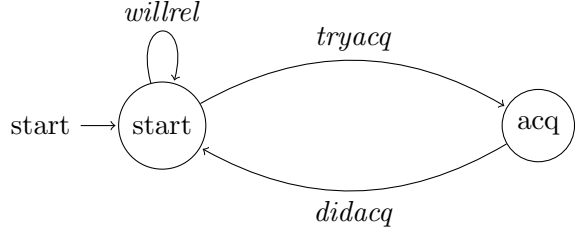
\begin{figure}
  \centering
  \begin{tikzpicture}[shorten >=1pt,node distance=5cm,on grid,auto]
    \node[state,initial] (start) {start};
    \node[state] (acquiring) [right=of start] {acq};
    \path[->]
      (start) edge [loop above] node {$\mathit{willrel}$} (start)
      (start) edge [bend left] node {$\mathit{tryacq}$} (acquiring)
      (acquiring) edge [bend left] node {$\mathit{didacq}$} (start);
  \end{tikzpicture}
  \caption{A DFA representation of the two phases in the deadlock monitor protocol. Non-branching messages are omitted for presentation purposes like in Figure~\ref{fig:dfa-auction}.} \label{fig:dfa-dd}
\end{figure}
\end{rem}

\subsection{Subsynchronizing Constraint}
\label{sec:phasing-ssync}
We note in \autoref{sec:background-shared} that in previous work~\cite{Balzer17icfp}, we require session types to be equi-synchronizing,
which requires that processes following the protocol are released at the exact type at which they were acquired. This constraint guarantees
that clients do not acquire at a type that they do not expect. With the introduction of subtyping however, there are two major relaxations
that we propose on this constraint.

\paragraph{Releasing at a subtype}
A client $P$ using some channel as some type $\St{a}{A}$ can safely communicate with any (shared) process offering a channel of type $\St{a}{A'}$ such that $\So{A'} \leq \So{A}$ due to subtyping. If another client acquires $\So{a}$ and releases it at some $\So{A''}$ such that $\So{A''} \leq \So{A'}$, then $P$ can still safely communicate along $\So{a}$ since $\So{A''} \leq \So{A}$ by transitivity. Thus, one reasonable relaxation to the equi-synchronizing constraint is that processes do not need to be released at the same exact type but instead a subtype.

\paragraph{Branches that never occur}
A major consequence of subtyping is that providers and clients can wait on some branches in the internal and external choices which in fact never will be sent by the other party. For example, suppose a provider $P$ provides a channel of type $\So{A} = {\Lupls\extch{a:\downsl{A}, b:\downsl{B}}}$. Assuming some unrelated $\So{B}$, we can see that $\So{A}$ is not equi-synchronizing because the $b$ branch can lead to releasing at a different type. However, suppose some client $C$ views the channel as  ${\Lupls\extch{a:\downsl{A}}}$ -- in this case, $P$ can only receive $a$, and the $b$ branch can safely be ignored since $C$ will never send the $b$ label. This points to the necessity of using both the provider and client types to more finely verify the synchronizing constraint. Of course, if there is another client $D$ that views the channel in a way that the $b$ branch can be taken, then the entire setup is not synchronizing. Thus, we must verify the synchronization constraint for all pairs of providers and clients.

Following previous work~\cite{Balzer17icfp}, we formulate constraints by extending the shared types: ${\hat{A} \defeq \bot \defor \So{A} \defor \top}$ where $\bot \leq \So{A} \leq \top$ for any $\So{A}$. Intuitively, $\top$ indicates a channel that has not been acquired yet (no constraints on a future release), $\So{A}$ indicates the previous presentation of shared channels, and $\bot$ indicates a channel that will never be available (hence, any client attempting to acquire from this channel will never succeed and be blocked).

We are now ready to present the \emph{subsynchronizing} judgment, interpreted coinductively, which is of the form $\dsync{A}{B}{\hat{D}}$ for some $A$ and $B$ such that $A \leq B$. It asserts that a provider providing a channel of type $A$ and a client using that channel with type $B$ is {subsynchronizing} with respect to some constraint $\hat{D}$. To verify a pair of types $A$ and $B$ to be subsynchronizing, we take $\top$ as its initial constraint (recall that $\top$ represents no constraint), that is, we say that $A$ and $B$ are subsynchronizing if $\dsync{A}{B}{\top}$.

\begin{small}
  \[
    \infer[S1]{\dsync{1}{1}{\hat{D}}}{}
  \]
  \[
    \infer[S{\otimes}]{\dsync{\tensor{A}{B}}{\tensor{A'}{B'}}{\hat{D}}}
    {\dsyncl{B}{B'}{\hat{D}}}
    \quad
    \infer[S{\multimap}]{\dsync{\loli{A}{B}}{\loli{A'}{B'}}{\hat{D}}}
    {\dsyncl{B}{B'}{\hat{D}}}
  \]
  \[
    \infer[S{\oplus}]{\dsync{\intch{\caselist{l}{A}}}
    {\intch{\caselist{l}{A'}, \caselist{m}{B}}}{\hat{D}}}
    {\forall i \in \overline{l} \quad \dsyncl{{A_i}}{{A_i'}}{\hat{D}}}
    \quad
    \infer[S{\&}]{\dsync{\extch{\caselist{l}{A}, \caselist{m}{B}}}
    {\extch{\caselist{l}{A'}}}{\hat{D}}}
    {\forall i \in \overline{l} \quad \dsyncl{{A_i}}{{A_i'}}{\hat{D}}}
  \]
  \[
    \infer[S\Lupll]{\dsync{\upll{A}}{\upll{A'}}{\hat{D}}}{\dsyncl{A}{A'}{\hat{D}}}
    \quad
    \infer[S\Ldownll]{\dsync{\downll{A}}{\downll{A'}}{\hat{D}}}{\dsyncl{A}{A'}{\hat{D}}}
  \]
  \[
    \infer[S\Lupls]{\dsync{\upls{A}}{\upls{A'}}{\top}}{\dsyncl{A}{A'}{\upls{A}}}
    \quad
    \infer[S\Ldownsl]{\dsync{\downsl{A}}{\downsl{A'}}{\hat{D}}}{\dsyncs{A}{A'}{\top} & \downsl{A} \leq \hat{D}}
  \]
  \[
    \infer[S\Lupls\Lupll]{\dsync{\upls{A}}{\upll{A'}}{\top}}{\dsyncl{A}{A'}{\upls{A}}}
    \quad
    \infer[S\Ldownsl\Ldownll]{\dsync{\downsl{A}}{\downll{A'}}{\hat{D}}}{\dsync{\So{A}}{\Lo{A'}}{\top} & \downsl{A} \leq \hat{D}}
  \]
\end{small}%

The general progression of derivations to verify that two types are subsynchronizing is to first look for an upshift $\Lupls$ on the provider's type, involving either $S\Lupls$ or $S\Lupls\Lupll$. After encountering a $\Lupls$, it ``records'' the provider's type as the constraint and continues to look at the continuations of the types. When encountering internal and external choices, it only requires the continuations for the common branches to be subsynchronizing. When it encounters a downshift $\Ldownsl$ from the provider's side, it checks if the release point as denoted by the continuation of $\Ldownsl$ is a subtype of the recorded constraint, in which case it continues with the derivation with the $\top$ constraint.

\begin{rem}
	Subsynchronizing constraint is a generalization of the equi-synchronizing constraint. In particular, if $A$ is equi-synchronizing, then the pair $A, A$ are subsynchronizing and vice versa.
\end{rem}

\section{Metatheory}
\label{sec:metatheory}
In this section we present $\sillSLeq$, a message-passing concurrency system
implementing the subtyping that we propose along with progress and preservation
theorems.

\subsection{Process Typing}
\label{ssec:process_typing}
We take the typing judgment presented in \autoref{sec:background-linear} and extend it with shared channels as introduced in \autoref{sec:background-shared}:
\begin{align*}
	\JTs{\Gamma}{P}{a}{A} \\
	\JTl{\Gamma}{\Delta}{Q}{a}{A}
\end{align*}
where $\Gamma = \Sto{{a_1}}{\hat{A_1}}, \ldots, \Sto{{a_n}}{\hat{A_n}}$ is a structural context of shared channels and constraints ($\bot$ and $\top$) which can appear at runtime.

The first judgment asserts that a process term $P$ provides a shared channel $\St{a}{A}$ while using shared channels in $\Gamma$; the lack of dependence on any linear channels $\Delta$ is due to the \emph{independence principle} presented in~\cite{Balzer17icfp}. The second judgment asserts that $Q$ provides a linear channel $\Lt{a}{A}$ while using shared channels in $\Gamma$ and linear channels in $\Delta$.

\paragraph{Global signature}
In the following sections, we will implicitly assume a global signature $\Sigma$, which is a set of process definitions that can be thought as the process calculi analogue to a signature consisting of function definitions. A process definition consists of the offering channel name and its type, the client channel names and their types, and the process term:

\begin{align*}
  \Sigma \defeq &\cdot \defor \Sigma, \spawn{\Lto{x}{\Lo{A}}}{\Lo{X}}{\Ltlist{y}{B},\Stlist{w}{E}} = P \\
  \defor &\Sigma , \spawn{\Sto{z}{\So{C}}}{\So{Z}}{\Stlist{v}{D}} = Q
\end{align*}
Leaving aside the $\cdot$ which denotes an empty signature, the former denotes a linear process definition of a process named $\Lo{X}$ that offers a channel $\Lt{x}{A}$ while using linear channels $\Lt{{y_1}}{{B_1}}, \ldots, \Lt{{y_n}}{{B_n}}$ and shared channels $\St{{w_1}}{{E_1}}, \ldots, \St{{w_m}}{{E_m}}$ for some $n$ and $m$, where $P$ consists of its implementation. Similarly, the latter denotes a shared process definition of a process named $\So{Z}$ that offers a channel $\St{z}{C}$ while using shared channels $\St{{v_1}}{{D_1}}, \ldots, \St{{v_n}}{{D_n}}$ for some $n$, where $Q$ consists of its implementation. Again, it is important that shared process definitions do not depend on linear channels due to the independence principle.

\subsubsection{Identity Rules}
\emph{Forwarding} is a fundamental operation that allows a process to identify its offering channel with a channel it uses if the types are compatible.

\begin{small}
  \[
    \infer[\Lo{ID}]{\JTl{\Gamma}{\Lt{y}{B}}{\fwdll{x}{y}}{x}{A}}{\Lo{B} \leq \Lo{A}}
    \quad
    \infer[\So{ID}]{\JTs{\Gamma, \Sto{y}{\hat{B}}}{\fwdss{x}{y}}{x}{A}}{\hat{B} \leq \So{A}}
  \]
  \[
    \infer[ID_{\scriptscriptstyle {LS}}]{\JTl{\Gamma, \Sto{y}{\hat{B}}}{\cdot}{\fwdls{x}{y}}{x}{A}}{\hat{B} \leq \Lo{A}}
  \]
\end{small}%

The rules $\Lo{ID}$ and $\So{ID}$ require the offering channel to be a supertype of the channel it is being identified with. Since we syntactically distinguish shared channels and linear channels, we require an additional rule $ID_{\scriptscriptstyle {LS}}$ that allows linear channels to be forwarded with a shared channel if the subtyping relation holds.
	
When a linear process \emph{spawns} another linear process, it can transfer channels that it currently communicates with to the new process.
In $\sillS$, this resulted in linear to linear and shared to shared channel substitutions, but with subtyping, the rule must now divide the
channel substitutions into three parts: linear to linear substitutions, shared to linear substitutions, and shared to shared substitutions.
The shared to linear substitution in particular occurs when a process definition expects a linear channel (or some type $\Lupll \ldots$) and
is instead given a smaller shared channel, and is in fact the key to the expressiveness of our system.
\[
\infer[SP_{\scriptscriptstyle {LL}}]
{\JTl{\Gamma}{\Delta, \Ltlist{y}{B}}
	{\spawn{\Lo{x}}{\Lo{X}}{\Lolist{y}, \Solist{v}, \Solist{w}}; Q}{z}{C}
}{\substack{
		\Stolist{v}{\hat{D}} \in \Gamma \\
		\Stolist{w}{\hat{E}} \in \Gamma
	}&\substack{
		\Lolist{B} \leq \Lolist{B'} \\
		\overline{\hat{D}} \leq \Lolist{D'} \\
		\overline{\hat{E}} \leq \Solist{E'}
	}&\substack{
          \left(
    	  \spawn{\Lto{x'}{\Lo{A}}}{X_L}{\Ltlist{y'}{B'},\Ltlist{v'}{D'},\Stlist{w'}{E'}} = P
	  \right) \in \Sigma \\
        \JTl{\Gamma}{\Delta, \Lt{x}{A}}{Q}{z}{C}}
}
\]

Similar to forwarding, there are two additional spawn rules (linear to shared and shared to shared) due to the syntactical distinguishment of the two modalities:
\[
  \infer[SP_{\scriptscriptstyle {LS}}]
  {\JTl{\Gamma}{\Delta}{\spawn{\So{x}}{\So{X}}{\Solist{y}}; Q}{z}{C}
    }{\substack{
      \Stolist{y}{\hat{B}} \in \Gamma \\
      \overline{\hat{B}} \leq \Solist{B'}
      }& \left(
      \spawn{\Sto{x'}{\So{A}}}{X_S}{\Stlist{y'}{B'}} = P
    \right) \in \Sigma &
    \JTl{\Gamma, \St{x}{A}}{\Delta}{Q}{z}{C}
  }
\]
\[
  \infer[SP_{\scriptscriptstyle {SS}}]
  {\JTs{\Gamma}{\spawn{\So{x}}{\So{X}}{\Solist{y}}; Q}{z}{C}
    }{\substack{
      \Stolist{y}{\hat{B}} \in \Gamma \\
      \overline{\hat{B}} \leq \Solist{B'}
      }& \left(
      \spawn{\Sto{x'}{\So{A}}}{X_S}{\Stlist{y'}{B'}} = P
    \right) \in \Sigma &
    \JTs{\Gamma, \St{x}{A}}{Q}{z}{C}
  }
\]

\subsubsection{Logical Rules}

As in standard sequent calculus presentations, typing judgments involving connectives are presented through left and right rules. The multiplicative unit $1$ denotes termination; providers must \emph{close} their offering channel while clients must \emph{wait} for the channel to close:
\begin{small}
  \[
    \infer[1L]{\JTl{\Gamma}{\Delta, \Lto{x}{1}}{\wait{x};P}{z}{C}}
    {\JTl{\Gamma}{\Delta}{P}{z}{C}}
    \quad
    \infer[1R]{\JTlr{\Gamma}{\cdot}{\close{x}}{x}{1}}{}
  \]
\end{small}%

For tensor ($\tensor{A}{B}$), providers (right rule) must \emph{send} a channel of some type $C_*$ such that $C_* \leq \Lo{A}$ (note that
$C_*$ can be either shared or linear, meaning there must be a rule covering each case separately). On the other hand, clients (left rule)
must \emph{receive} a channel of type $\Lo{A}$ (which due to subtyping could be smaller in actuality).
\begin{small}
  \[
    \infer[{\otimes} L]{\JTl{\Gamma}{\Delta, \Lto{x}{\tensor{A}{B}}}{\asgn{\Lo{y}}{\recv{x}}; P}{z}{C}}
    {\JTl{\Gamma}{\Delta, \Lt{x}{B}, \Lt{y}{A}}{P}{z}{C}}
    \;
    \infer[{\otimes} R]{\JTlr{\Gamma}{\Delta, \Lt{y}{A'}}{\sendl{x}{y}; P}{x}{\tensor{A}{B}}}
    {\Lo{A'} \leq \Lo{A} &
    \JTl{\Gamma}{\Delta}{P}{x}{B}}
  \]
  \[
    \infer[{\otimes} \So{R}]{\JTlr{\Gamma, \Sto{y}{\hat{A}}}{\Delta}{\sends{x}{y}; P}{x}{\tensor{A}{B}}}
    {\hat{A} \leq \Lo{A} &
    \JTl{\Gamma, \Sto{y}{\hat{A}}}{\Delta}{P}{x}{B}}
  \]
\end{small}%

Dually for linear implication ($\loli{A}{B}$), clients must \emph{send} a channel of some subtype of $\Lo{A}$ while providers must \emph{receive} a channel of type $\Lo{A}$:
\[
  \infer[{\multimap} L]{\JTl{\Gamma}{\Delta, \Lto{x}{\loli{A}{B}}, \Lt{y}{A'}}{\sendl{x}{y}; P}{z}{C}}
  {\Lo{A'} \leq \Lo{A} &
  \JTl{\Gamma}{\Delta, \Lto{x}{B}}{P}{z}{C}}
\]
\[
  \infer[{\multimap} R]{\JTlr{\Gamma}{\Delta}{\asgn{\Lo{y}}{\recv{x}}; P}{x}{\loli{A}{B}}}
  {\JTl{\Gamma}{\Delta, \Lt{y}{A}}{P}{x}{B}}
\]
\[
  \infer[{\multimap} \So{L}]{\JTl{\Gamma, \Sto{y}{\hat{A}}}{\Delta, \Lto{x}{\loli{A}{B}}}{\sends{x}{y}; P}{z}{C}}
  {\hat{A} \leq \Lo{A} &
  \JTl{\Gamma, \Sto{y}{\hat{A}}}{\Delta, \Lto{x}{B}}{P}{z}{C}}
\]

In this system, binary internal and external choices, $\Lo{A} \oplus \Lo{B}$ and $\Lo{C} \& \Lo{D},$ are generalized to their $n$-ary versions, $\intch{\caselist{l}{A}}$ and $\extch{\caselist{m}{B}}$, where each continuation type $A_i$ or $B_i$ has a corresponding (unique) label $l_i$ or $m_i$. For internal choice, providers must \emph{send} a label $l_i$ and then continue as $A_i$ whereas clients must \emph{receive} a label and continue as the type that correspond with the label it received.
\begin{small}
  \[
    \infer[{\oplus} L]{\JTl{\Gamma}{\Delta, \Lto{x}{\intch{\caselist{l}{A}}}}{\casep{x}{\caselistp{l}{P}}}{c}{Z}}
    {\forall i \in \overline{l} \quad \JTl{\Gamma}{\Delta, \Lt{x}{{A_i}}}{P_i}{c}{Z}}
    \quad
    \infer[{\oplus} R]{\JTlr{\Gamma}{\Delta}{x.i; P}{x}{\intch{\caselist{l}{A}}}}
    {i \in \overline{l} & \JTl{\Gamma}{\Delta}{P}{x}{{A_i}}}
  \]
\end{small}%

Dually for external choice, clients \emph{send} a label whereas providers \emph{receive} and branch on the input label:
\begin{small}
  \[
    \infer[{\&} L]{\JTl{\Gamma}{\Delta, \Lto{x}{\extch{\caselist{l}{A}}}}{x.i; P}{z}{C}}
    {i \in \overline{l} & \JTl{\Gamma}{\Delta, \Lt{x}{{A_i}}}{P}{z}{{C}}}
    \quad
    \infer[{\&} R]{\JTlr{\Gamma}{\Delta}{\casep{x}{\caselistp{l}{P}}}{x}{\extch{\caselist{l}{A}}}}
    {\forall i \in \overline{l} \quad \JTl{\Gamma}{\Delta}{P_i}{x}{{A_i}}}
  \]
\end{small}%

Next, $\upls{A}$ signifies a synchronization point where clients must \emph{acquire} while (shared) providers must \emph{accept} a client, both proceeding with $\Lo{A}$ as the continuation.
\begin{small}
  \[
    \infer[\Lupls L]{\JTl{\Gamma, \Sto{x}{\hat{A}}}{\Delta}{\asgn{\Lo{x}}{\acqs{x}}; P}{z}{C}}
    {\hat{A} \leq \upls{A} & \JTl{\Gamma, \Sto{x}{\hat{A}}}{\Delta, \Lt{x}{A}}{P}{z}{C}}
    \quad
    \infer[\Lupls R]{\JTsr{\Gamma}{\asgn{\Lo{x}}{\accs{x}}; P}{x}{\upls{A}}}
    {\JTl{\Gamma}{\cdot}{P}{x}{A}}
  \]
\end{small}%

$\downsl{A}$ signifies a point where clients must \emph{release} while providers \emph{detach} from a linear session, returning to a shared state ready to \emph{accept} another client.
\begin{small}
  \[
    \infer[\Ldownsl L]{\JTl{\Gamma}{\Delta, \Lto{x}{\downsl{A}}}{\asgn{\So{x}}{\rels{x}}; P}{z}{C}}
    {\JTl{\Gamma, \St{x}{A}}{\Delta}{P}{z}{C}}
    \quad
    \infer[\Ldownsl R]{\JTlr{\Gamma}{\cdot}{\asgn{\So{x}}{\dets{x}}; P}{x}{\downsl{A}}}
    {\JTs{\Gamma}{P}{x}{A}}
  \]
\end{small}%

Finally, we require the linear variants of the up and downshifts, which by themselves can be interpreted as synchronization points in a
linear protocol~\cite{Pfenning15fossacs}. However, in this paper, their purpose is to safely act as supertypes to corresponding shared up
and downshifts, which allow linearity to be enforced on clients in shared protocols.
\begin{small}
  \[
    \infer[\Lupll L]{\JTl{\Gamma}{\Delta, \Lto{x}{\upll{A}}}{\asgn{\Lo{y}}{\acql{x}}; P}{z}{C}}
    {\JTl{\Gamma}{\Delta, \Lt{y}{A}}{P}{z}{C}}
    \quad
    \infer[\Lupll R]{\JTlr{\Gamma}{\Delta}{\asgn{\Lo{y}}{\accl{x}}; P}{x}{\upll{A}}}
    {\JTl{\Gamma}{\Delta}{P}{x}{A}}
  \]
  \vspace{-1mm}
  \[
    \infer[\Ldownll L]{\JTl{\Gamma}{\Delta, \Lto{x}{\downll{A}}}{\asgn{\Lo{y}}{\rell{x}}; P}{z}{C}}
    {\JTl{\Gamma}{\Delta, \Lt{y}{A}}{P}{z}{C}}
    \quad
    \infer[\Ldownll R]{\JTlr{\Gamma}{\Delta}{\asgn{\Lo{y}}{\detl{x}}; P}{x}{\downll{A}}}
    {\JTl{\Gamma}{\Delta}{P}{y}{A}}
  \]
\end{small}

One important observation is that typing judgments remain local in the presence of subtyping; the channels in $\Gamma$ and $\Delta$ may be provided by processes at some subtype (maintained in the configuration; see \autoref{sec:metatheory-cfg-typing}) and need not match. We therefore do not adopt a general subsumption rule that allows arbitrary substitutions that preserve subtyping and instead precisely manage where subtyping occurs in the system.

\subsubsection{Structural Rules}
Structural rules are kept implicit in the system, but informally, the linear context $\Delta$ only allows exchange whereas the shared context $\Gamma$ allows all structural rules.

\subsection{Dynamics}
The operational semantics of the system is formulated through \emph{multiset rewriting rules}~\cite{Cervesato09ic}, which is of form ${S_1,
\ldots, S_n \to T_1, \ldots, T_m}$, where each $S_i$ and $T_j$ corresponds to a \emph{process predicate}, which captures the state of a
particular process and is of form: $$S \defeq \procs{a}{P} \defor \unavail{b} \defor \procl{c}{Q} \defor \connect{d}{e} \defor \pdef{A}$$
where $P$ and $Q$ are process terms as formulated in \autoref{ssec:process_typing}. The predicates $\procs{a}{P}$ and $\procl{c}{Q}$ denote
shared and linear processes that offer channels along $\So{a}$ and $\Lo{c}$ while executing process terms $P$ and $Q$, respectively.  The
predicate $\unavail{b}$ denotes a shared process that is currently unavailable, for example due to it being acquired by another client, and
the predicate $\connect{d}{e}$ is an explicit predicate that connects a shared channel with a linear channel which is needed to dynamically
express shared to linear subtyping. Finally, $\pdef{A}$ is a (persistent) linear or shared process definition as demonstrated in $\Sigma$.
We adopt $\Psi_a$ as a metavariable for some linear process predicate offering $\Lo{a}$; that is, $\Psi_a$ is either $\procl{a}{P}$ for some
$P$ or $\connect{a}{b}$ for some $\So{b}$.

Each multiset rule captures local transitions in the system; for example, there are three rules that represent forwarding, each corresponding to the appropriate forwarding typing judgments:
\begin{small}
  \begin{gather*}
    \tag{D-FWDLL} \label{dyn:fwdll}
    \procl{a}{\fwdll{a}{b}} \to \cdot \quad (\Lo{b} := \Lo{a}, \So{b} := \So{a}) \\
    \tag{D-FWDSS} \label{dyn:fwdss}
    \procs{a}{\fwdss{a}{b}} \to \unavail{a} \quad (\So{b} := \So{a}) \\
    \tag{D-FWDLS} \label{dyn:fwdls}
    \procl{a}{\fwdls{a}{b}} \to \connect{a}{b}
  \end{gather*}
\end{small}%
The rules \ref{dyn:fwdll} and \ref{dyn:fwdss} are two exceptions to the local transformations; they require the two channels to be
``globally'' identified. The rule \ref{dyn:fwdls} says that a linear process that forwards with a shared channel must transition to a connect predicate, which serves as a placeholder to denote shared to linear subtyping.

A linear to linear spawn creates a process $P$ offering a fresh channel $\Lo{c}$. One important point is that fresh linear channels $\overline{\Lo{d'}}$ are allocated alongside corresponding connect predicates due to the possibility of shared channels $\overline{\So{d}}$ being ``passed'' to the new process as linear channels.
\begin{gather*}
  \tag{D-SPAWNLL} \label{dyn:spawnll}
  \substack{
    \procl{a}{\spawn{\Lo{x}}{\Lo{X}}{\Lolist{b}, \Solist{d}, \Solist{e}}; Q} \\
    \pdef{(\spawn{\Lto{x'}{\Lo{A}}}{X_L}{\Ltlist{y'}{B'},\Ltlist{v'}{D'},\Stlist{w'}{E'}})=P}
  }
  \to
  \substack{
    \procl{a}{[\Lsub{c}{x}]Q},
    \procl{c}{[\Lsub{c}{x'}, \Lsublist{b}{y'}, \Lsublist{d'}{v'}, \Ssublist{e}{w'}]P} \\
  \overline{\connect{d'}{d}, \unavail{d'}} \fresh{\overline{d'}, c}} \\
\end{gather*}
Note that corresponding $\unavail{d'}$ predicates are spawned which solely makes later proofs easier. These can essentially be ignored for now.

The two other spawn cases are similar, except since since linear channels cannot be passed to shared processes, the verbose allocation of connect predicates are not necessary.
\begin{gather*}
  \tag{D-SPAWNLS} \label{dyn:spawnls}
  \substack{
    \procl{a}{\spawn{\So{x}}{\So{X}}{\Solist{b}}; Q} \\
    \pdef{(\spawn{\Sto{x'}{\So{A}}}{\So{X}}{\Stlist{y'}{B'}})=P}
  }\to
  \procl{a}{[\Ssub{c}{x}]Q}, \procs{c}{[\Ssub{c}{x'}, \Ssublist{b}{y'}]P} \fresh{c} \\
  \tag{D-SPAWNSS} \label{dyn:spawnss}
  \substack{
    \procs{a}{\spawn{\So{x}}{\So{X}}{\Solist{b}}; Q} \\
    \pdef{(\spawn{\Sto{x'}{\So{A}}}{\So{X}}{\Stlist{y'}{B'}})=P}
  }
  \to
  \procs{a}{[\Ssub{c}{x}]Q}, \procs{c}{[\Ssub{c}{x'}, \Ssublist{b}{y'}]P} \fresh{c}
\end{gather*}

For the unit $1$, a client \emph{waiting} for a channel to \emph{close} can proceed when the corresponding provider \emph{closes} its channel.
\begin{small}
  \begin{gather*}
    \tag{D-$1$} \label{dyn:unit}
    \procl{a}{\wait{b};P}, \procl{b}{\close{b}} \to \procl{a}{P}
  \end{gather*}
\end{small}

The left hand side of the dynamics follow a pattern where one process receives while another process sends. Starting with $\otimes$ and $\multimap$:
\begin{small}
  \begin{gather*}
    \tag{D-$\otimes$} \label{dyn:tensor}
    \procl{a}{\asgn{\Lo{y}}{\recv{b}}; P}, \procl{b}{\sendl{b}{c}; Q}, \Psi_c
    \to
    \procl{a}{[\Lsub{c}{y}]P}, \procl{b}{Q}, \Psi_c \\
    \tag{D-$\multimap$} \label{dyn:loli}
    \procl{a}{\sendl{b}{c}; P}, \procl{b}{\asgn{\Lo{y}}{\recv{b}}; Q}, \Psi_c
    \to
    \procl{a}{P}, \procl{b}{[\Lsub{c}{y}]Q}, \Psi_c
  \end{gather*}
\end{small}

When shared channels are sent instead, a fresh channel $d$ is allocated and a connect predicate connects the shared channel:
\begin{small}
  \begin{align*}
    \tag{D-$\otimes$2} \label{dyn:tensor2}
    &\procl{a}{\asgn{\Lo{y}}{\recv{b}}; P}, \procl{b}{\sends{b}{c}; Q} \\
    \to \quad
    &\procl{a}{[\Lsub{d}{y}]P}, \procl{b}{Q}, \connect{d}{c}, \unavail{d} \fresh{d} \\\\
    \tag{D-$\multimap$2} \label{dyn:loli2}
    &\procl{a}{\sends{b}{c}; P}, \procl{b}{\asgn{\Lo{y}}{\recv{b}}; Q} \\
    \to \quad
    &\procl{a}{P}\procl{b}{[\Lsub{d}{y}]Q}, \connect{d}{c}, \unavail{d} \fresh{d}
  \end{align*}
\end{small}%

For $\oplus$ and $\&$, the pattern of one side sending (a label) and the other receiving is maintained:
\begin{small}
  \begin{gather*}
    \tag{D-$\oplus$} \label{dyn:intch}
    \procl{a}{\casep{b}{\caselistp{l}{P}, \caselistp{m}{P}}}, \procl{b}{b.i; Q}
    \to
    \procl{a}{P_i}, \procl{b}{Q} \quad (i \in \overline{l}) \\
    \tag{D-$\&$} \label{dyn:extch}
    \procl{a}{b.i; P}, \procl{b}{\casep{b}{\caselistp{l}{Q}, \caselistp{m}{Q}}}
    \to
    \procl{a}{P}, \procl{b}{Q_i} \quad (i \in \overline{l})
  \end{gather*}
\end{small}%
An important point is that due to subtyping, the process receiving a label can accept a superset of the labels that the process sending will send. This is syntactically expressed by having the recipient case on the list $\overline{l}, \overline{m}$ while having the sender pick a label in $\overline{l}$.

Now for the modal connectives, the idea is similar to the previous logical connectives; for $\Lupls$, a client must \emph{acquire} a shared channel and the corresponding shard provider must \emph{accept}. Similarly for $\Ldownsl$, a client must \emph{release} while the provider must \emph{detach}, returning to a shared process:
  \begin{gather*}
    \tag{D-$\Lupls$} \label{dyn:upls}
    \procl{a}{\asgn{\Lo{x}}{\acqs{b}}; P}, \procs{b}{\asgn{\Lo{x}}{\accs{b}}; Q}
    \to
    \substack{
    \procl{a}{[\Lsub{b}{x}]P}, \procl{b}{[\Lsub{b}{x}]Q},
    \\ \unavail{b}} \\
    \tag{D-$\Ldownsl$} \label{dyn:downsl}
    \substack{\procl{a}{\asgn{\So{x}}{\rels{b}}; P}, \procl{b}{\asgn{\So{x}}{\dets{b}}; Q},
    \\ \unavail{b}}
    \to
    \procl{a}{[\Ssub{b}{x}]P}, \procs{b}{[\Ssub{b}{x}]Q}
  \end{gather*}

The linear variants have a similar semantics:
\begin{small}
  \begin{gather*}
    \tag{D-$\Lupll$} \label{dyn:upll}
    \procl{a}{\asgn{\Lo{x}}{\acql{b}}; P}, \procl{b}{\asgn{\Lo{x}}{\accl{b}}; Q}
    \to
    \procl{a}{[\Lsub{b}{x}]P}, \procl{b}{[\Lsub{b}{x}]Q} \\
    \tag{D-$\Ldownll$} \label{dyn:downll}
    \procl{a}{\asgn{\Lo{x}}{\rell{b}}; P}, \procl{b}{\asgn{\Lo{x}}{\detl{b}}; Q}
    \to
    \procl{a}{[\Lsub{b}{x}]P}, \procl{b}{[\Lsub{b}{x}]Q}
  \end{gather*}
\end{small}%

Finally, when a client linearly \emph{acquires} what happens to be a shared process, it must go through the connect predicate, and similarly when a client linearly \emph{releases} a provider that is \emph{detaching} to a shared state, a connect predicate is allocated:
\begin{gather*}
  \tag{D-$\Lupls$2} \label{dyn:upls2}
  \substack{
    \procl{a}{\asgn{\Lo{x}}{\acql{b}}; P}, \connect{b}{c} \\
    \procs{c}{\asgn{\Lo{x}}{\accs{c}}; Q}
  }\to
  \substack{
    \procl{a}{[\Lsub{c}{x}]P}, \procl{c}{[\Lsub{c}{x}]Q} \\ 
  \unavail{c}} \\
  \tag{D-$\Ldownsl$2} \label{dyn:downsl2}
  \substack{
    \procl{a}{\asgn{\Lo{x}}{\rell{c}}; P}, \procl{c}{\asgn{\So{x}}{\dets{c}}; Q}, \\
  \unavail{c}}
  \to
  \substack{
    \procl{a}{[\Lsub{b}{x}]P}, \connect{b}{c}, \\ 
  \unavail{b}}
\end{gather*}

\subsection{Processes and Configuration}
A configuration consists of a list of shared process predicates $\Lambda$ and a list of linear process predicates $\Theta$. The order of shared processes have
no structure, but the order of linear processes can be seen to form a tree structure; a linear process can use channels offered by processes to its right, and
due to linearity, if it is using a channel, it must be the unique process doing so.
\begin{align*}
  \Omega &\defeq \Lambda;\Theta \\
  \Lambda &\defeq \cdot \defor \Lambda_1, \Lambda_2 \defor \procs{a}{P} \defor \unavail{a} \\
  \Theta &\defeq \cdot \defor \procl{a}{P}, \Theta' \defor \connect{a}{b}, \Theta' \\
\end{align*}

\paragraph{Well-formedness}
$\Lambda$ is well-formed if for any channel name $a$, ${\procs{a}{P}, \unavail{a} \notin \Lambda}$.  Similarly, $\Theta$ is well-formed if for any $a$, $\Psi_a,
\Psi_a' \notin \Theta$ where $\Psi_a \neq \Psi_a'$. The configuration $\Lambda;\Theta$ is well-formed if both its fragments are well-formed and
${\Psi_a \in \Theta \to \unavail{a} \in \Lambda}$.

\subsection{Configuration Typing}
\label{sec:metatheory-cfg-typing}
A well-formed configuration $\Lambda;\Theta$ is typed by its shared and linear fragments.
\begin{small}
\[
  \infer[\Omega]{\JCb{\Gamma}{\Lambda}{\Theta}{\Delta}}{\JCs{\Gamma}{\Lambda}{\Gamma} & \JCl{\Gamma}{\Theta}{\Delta}}
\]
\end{small}
\begin{small}
\[
  \infer[\Lambda 1]{\JCs{\Gamma}{\cdot}{\cdot}}{}
  \quad
  \infer[\Lambda 2]{\JCs{\Gamma}{\Lambda_1, \Lambda_2}{\Gamma_1, \Gamma_2}}
    {\JCs{\Gamma}{\Lambda_1}{\Gamma_1} & \JCs{\Gamma}{\Lambda_2}{\Gamma_2}}
\]
\[
  \infer[\Lambda 3]{\JCs{\Gamma}{\procs{a}{P}}{\St{a}{A}}}{\dsyncs{A'}{A}{\top} & \JTs{\Gamma}{P}{a}{A'}}
  \quad
  \infer[\Lambda 4]{\JCs{\Gamma}{\unavail{a}}{\Sto{a}{\hat{A}}}}{}
\]
\end{small}
\begin{small}
\[
  \infer[\Theta 1]{\JCl{\Gamma}{\cdot}{\cdot}}{}
  \quad
  \infer[\Theta 2]{\JCpc{\Gamma}{a}{b}{\Lo{A}}{\Theta'}{\Delta'}}{\Sto{b}{\hat{B}} \in \Gamma & \So{b} \leq \Lo{A} &\JCl{\Gamma}{\Theta'}{\Delta'}}
\]
\[
  \infer[\Theta 3]{\JCpp{\Gamma}{a}{P}{\Lo{A}}{\Theta'}{\Delta'}}
    {\Sto{a}{\hat{A}} \in \Gamma & \dsyncl{A'}{A}{\hat{A}} & \JTl{\Gamma}{\Delta_a}{P}{a}{A'} & \JCl{\Gamma}{\Theta'}{\Delta_a, \Delta'}}
\]
\end{small}

\subsection{Lemmas}
In this section we present lemmas of interest to be used in the progress and preservation proofs. The proofs of each lemma are in \autoref{app:lemma-proofs}.

\subsubsection{Lemmas involving the Configuration}
\autoref{lem:cfg-peel} allows the tail of linear configurations to be peeled off, \autoref{lem:cfg-balance} asserts that an active shared
process prevents an active linear process of the same channel name, \autoref{lem:cfg-permutation} allows individual process predicates in
linear configurations to be moved around as long as the overall invariant that linear processes can only depend on processes to its right is
maintained, \autoref{lem:cfg-subs} allows the substitution of subconfigurations in a linear configuration if signatures match, and finally,
\autoref{lem:cfg-bigger} allows offering channels of linear processes to be viewed at supertypes.

\begin{lem}
  If $\JCl{\Gamma}{\Psi, \Theta}{\Delta}$, then $\JCl{\Gamma}{\Theta}{\Delta'}$ for some $\Delta'$. \\
  More generally, if $\JCl{\Gamma}{\Theta_1, \Theta_2}{\Delta}$, then $\JCl{\Gamma}{\Theta_2}{\Delta'}$ for some $\Delta'$.
  \label{lem:cfg-peel}
\end{lem}

\begin{lem}
  Given a well-formed $\Lambda;\Theta$, $\forall \procs{a}{-} \in \Lambda, \Psi_a \notin \Theta$
  \label{lem:cfg-balance}
\end{lem}

\begin{lem}
  If $\JCl{\Gamma}{\Psi_a, \Theta_1, \Psi_b, \Theta_2}{\Lt{a}{A}, \Delta}$ and $\Psi_a$ uses
  $\Lo{b}$, then
  $$\JCl{\Gamma}{\Psi_a, \Psi_b, \Theta_1, \Theta_2}{\Lt{a}{A}, \Delta}$$
  \label{lem:cfg-permutation}
\end{lem}

\begin{lem}
  If $\JCl{\Gamma}{\Psi, \Theta}{\Delta}$, $\JCl{\Gamma}{\Theta}{\Delta_p}$, and $\JCl{\Gamma}{\Theta'}{\Delta_p}$, then
  $$\JCl{\Gamma}{\Psi, \Theta'}{\Delta}$$  \\
  More generally, if $\JCl{\Gamma}{\Theta_1, \Theta_2}{\Delta}$, $\JCl{\Gamma}{\Theta_2}{\Delta_p}$, and $\JCl{\Gamma}{\Theta_2'}{\Delta_p}$, then
  $$\JCl{\Gamma}{\Theta_1, \Theta_2}{\Lt{a}{A}, \Delta}$$
  \label{lem:cfg-subs}
\end{lem}

\begin{lem}
  If $\JCl{\Gamma}{\Psi_a, \Theta'}{\Lt{a}{A'}, \Delta'}$, then for any $\Lo{B}$ such that $\Lo{A'} \leq \Lo{B}$,\\
  $\JCl{\Gamma}{\Psi_a, \Theta'}{\Lt{a}{B}, \Delta'}$.
  \label{lem:cfg-bigger}
\end{lem}

\subsubsection{Ordering of Contexts}
A linear context is smaller than another if it shares the same variables with their associated types respecting subtyping. Similarly, a shared context is smaller than another if it contains at least the same variables (could contain additional as shown in $\Gamma_{\preceq \cdot}$) with their associated types respecting subtyping.
\begin{small}
\[
  \infer[\Delta_{\leq \cdot}]{\cdot \leq \cdot}{} \quad
  \infer[\Delta_{\leq x}]{\Delta, \Lt{x}{A} \leq \Delta', \Lt{x}{A'}}{\Delta \leq \Delta' & \Lo{A} \leq \Lo{A'}}
\]
\[
  \infer[\Gamma_{\preceq \cdot}]{\Gamma \preceq \cdot}{} \quad
  \infer[\Gamma_{\preceq x}]{\Gamma, \Sto{x}{\hat{A}} \preceq \Gamma', \Sto{x}{\hat{A'}}}{\Gamma \preceq \Gamma' & \hat{A} \leq \hat{A'}}
\]
\end{small}%

The following two lemmas allow the substitution of smaller shared contexts in both the configuration typing and process typing judgments.

\begin{lem}
  Let $\Gamma' \preceq \Gamma$ and $\JTl{\Gamma}{\Delta}{P}{z}{C}$, then $\JTl{\Gamma'}{\Delta}{P}{z}{C}$.
  \label{lem:static-stable}
\end{lem}

\begin{lem}
  Let $\Gamma' \preceq \Gamma$ then
  \begin{enumerate}
    \item If $\JCl{\Gamma}{\Theta}{\Delta}$ for some $\Theta, \Delta$, then $\JCl{\Gamma'}{\Theta}{\Delta}$
    \item If $\JCs{\Gamma}{\Lambda}{\Gamma''}$ for some $\Lambda, \Gamma''$, then $\JCs{\Gamma'}{\Theta}{\Gamma''}$
  \end{enumerate}
  \label{lem:cfg-stable}
\end{lem}

\subsubsection{Subsynchronizing Judgment}
The following lemmas apply to the subsynchronizing judgment defined in \autoref{sec:phasing-ssync}.  \autoref{lem:dsync-bigger} allows the
client type (second argument) to become bigger, \autoref{lem:dsync-smaller} allows the provider type (first argument) to become smaller
under a specific circumstance, \autoref{lem:dsync-smaller-hat} allows the constraint (third argument) to become smaller if both provider and
clients are linear, and finally, \autoref{lem:subtype-meet} allows the construction of a smaller constraint given two subsynchronizing
judgments of the same provider and client types.

\begin{lem}
  If $A \leq B \leq C$ with all same modalities (that is, $A, B, C$ are either all linear or all shared) and $\dsync{A}{B}{\hat{D}}$, then $\dsync{A}{C}{\hat{D}}$ for some $\hat{D}$.
  \label{lem:dsync-bigger}
\end{lem}

\begin{lem}
  If $A \leq B \leq C$ with all same modalities, $\dsync{B}{C}{\hat{D}}$, and \\$\dsync{A}{C}{\hat{E}}$, then $\dsync{A}{C}{\hat{D}}$ for some $\hat{D}$ and $\hat{E}$.
  \label{lem:dsync-smaller}
\end{lem}

\begin{lem}
  If $\dsyncl{A}{B}{\hat{C}}$ and $\hat{D} \leq \hat{C}$, then $\dsyncl{A}{B}{\hat{D}}$ for some $\Lo{A}, \Lo{B}, \hat{C},$ and $\hat{D}$.
  \label{lem:dsync-smaller-hat}
\end{lem}

\begin{lem}
  If $\dsyncl{A}{B}{\hat{C}}$ and $\dsyncl{A}{B}{\hat{D}}$, ${\dsyncl{A}{B}{\hat{C} \land \hat{D}}}$ for some $\Lo{A}, \Lo{B}, \hat{C},$ and $\hat{D}$.
  \label{lem:dsync-meet}
\end{lem}
Note that the meet of two constraints $\hat{C} \land \hat{D}$ is defined in \autoref{lem:subtype-meet}.

\subsection{Theorems}
The preservation theorem, or session fidelity, guarantees that well-typed configurations remain well-typed. In particular, this means that processes will always adhere to the protocol denoted by the session type.
\begin{thm}[Preservation]
If $\JCb{\Gamma}{\Lambda}{\Theta}{\Delta}$ for some $\Lambda, \Theta, \Gamma,$ and $\Delta$, and $\Lambda;\Theta \rightarrow \Lambda';\Theta'$ for some $\Lambda';\Theta'$, then $\JCb{\Gamma'}{\Lambda'}{\Theta'}{\Delta}$ where $\Gamma' \preceq \Gamma$.
\end{thm}
Here, $\Gamma' \preceq \Gamma$ captures the idea that the configuration can gain additional shared processes and that the types of shared channels can become smaller. For example, if a process spawns an additional shared process, then the configuration will gain an additional channel in $\Gamma$ and if a shared channel is released to a smaller type, the type of the shared channel in $\Gamma$ can become smaller. Note that although it is indeed true that linear processes can be spawned, it will never appear in $\Delta$ since the linear channel that the newly spawned process offers must be consumed by the process that spawned the channel, meaning $\Delta$ is unchanged.
\begin{proof}
  By induction on the dynamics and constructing a well-typed (and therefore well-formed) configuration for each case. We present a simple
  case below; a complete proof is presented in~\autoref{app:preservation-proof}.
  \begin{case}
    \ref{dyn:fwdls}
    $$\procl{a}{\fwdls{a}{b}} \to \connect{a}{b}$$
    Let $\Psi_a = \procl{a}{\fwdls{a}{b}}$ and $\Psi_a' = \connect{a}{b}$.
    Let $\Theta = \Theta_1, \Psi_a, \Theta_2$. Then by well-formedness, $\Lambda = \unavail{a}, \Lambda_1$.
    \begin{align*}
        &\JCb{\Gamma}{\Lambda}{\Theta_1, \Psi_a, \Theta_2}{\Delta}	\tag{assumption} \\
        &\JCs{\Gamma}{\Lambda}{\Gamma} \quad \JCl{\Gamma}{\Theta_1, \Psi_a, \Theta_2}{\Delta}	\tag{by inversion on $\Omega$} \\
        &\JCl{\Gamma}{\procl{a}{\fwdls{a}{b}}, \Theta_2}{\Lt{a}{A}, \Delta_p} \tag{by Lemma~\ref{lem:cfg-peel} and expanding $\Psi_a$} \\
        &\JCl{\Gamma}{\Theta_2}{\Delta_p} \quad \JTl{\Gamma}{\cdot}{\fwdls{a}{b}}{a}{A'} \tag{by inversion on $\Theta 3$} \\
        &\Sto{b}{\hat{B}} \in \Gamma \quad \hat{B} \leq \Lo{A'} \tag{by inversion on $ID_{\scriptscriptstyle {LS}}$} \\
        &\hat{B} \leq \Lo{A} \tag{by transitivity of $\leq$} \\
        &\JCpc{\Gamma}{a}{b}{\Lo{A}}{\Theta_2}{\Delta_p} \tag{by $\Theta2$} \\
        &\JCl{\Gamma}{\Theta_1, \Psi_a', \Theta_2}{\Delta} \tag{by Lemma~\ref{lem:cfg-subs}} \\
        &\JCb{\Gamma}{\Lambda}{\Theta_1, \Psi_a', \Theta_2}{\Delta}	\tag{by $\Omega$}
    \end{align*}
    The well-formedness conditions are maintained because only $\Psi_a \in \Theta$ was replaced by $\Psi_a'$.
  \end{case}
  Many of the dynamics involving the standard logical connectives $({\otimes}, {\multimap}, {\oplus},$ and ${\&})$ follow a similar pattern
  and are fairly simple. However, cases involving the shift connectives $(\Lupls, \Ldownsl, \Lupll, \Ldownll)$ and linear to linear
  forwarding cause more complexities and require further subcase analysis. These cases are presented in detail
  in~\autoref{app:progress-proof}.
\end{proof}

The progress theorem is as in~\cite{Balzer17icfp}, where we only allow configurations to be stuck due to failure of some client to acquire, for example, due to
deadlock.
\begin{defi}
  A shared and linear process term $\text{proc}(a, P)$ is \emph{poised} if $P$ is currently communicating along its providing channel $a$. Poised process terms in $\sillSLeq$ are shown in the table below:
  \begin{center}
    \begin{tabular}{@{}ll@{}}
      Receiving & Sending \\
      \hline
                & $\procl{a}{\close{a}}$ \\
      $\procl{a}{\asgn{\Lo{x}}{\recv{a}}; P}$ & $\procl{a}{\sendl{a}{c}; P}$ \\
      $\procl{a}{\asgn{\So{x}}{\recv{a}}; P}$ & $\procl{a}{\sends{a}{c}; P}$ \\
      $\procl{a}{\casep{a}{\caselistp{l}{P}}}$ & $\procl{a}{a.i; P}$ \\
      $\procl{a}{\asgn{\Lo{x}}{\accl{a}}; P}$ & $\procl{a}{\asgn{\Lo{x}}{\detl{a}}; P}$ \\
      $\procs{a}{\asgn{\Lo{x}}{\accs{a}}; P}$ & $\procs{a}{\asgn{\Lo{x}}{\dets{a}}; P}$
    \end{tabular}
  \end{center}
  In particular, we say that a configuration is poised if all of its $\text{proc}(-, -)$ members are poised.
\end{defi}

\begin{thm}[Progress]
  If $\JCb{\Gamma}{\Lambda}{\Theta}{\Delta}$ then either:
  \begin{enumerate}
    \item $\Lambda;\Theta \rightarrow \Lambda';\Theta$ for some $\Lambda'$ or
    \item $\Lambda$ is poised and one of:
      \begin{enumerate}
        \item $\Lambda;\Theta \rightarrow \Lambda';\Theta'$ or
        \item $\Theta$ is poised or
        \item a linear process in $\Theta$ is stuck and therefore unable to acquire
      \end{enumerate}
  \end{enumerate}
\end{thm}
\begin{proof}
  For details, see \autoref{app:progress-proof}. We first show that either the shared configuration $\Lambda$ steps $(\Lambda \to \Lambda'$
  for some $\Lambda')$ or that $\Lambda$ is poised by induction on the derivation of $\JCs{\Gamma}{\Lambda}{\Gamma}$. If $\Lambda$ is
  poised, then we proceed by induction on the derivation of $\JCl{\Gamma}{\Theta}{\Delta}$ to show one of:
  \begin{enumerate}[label=(\alph*)]
    \item $\Lambda;\Theta \to \Lambda';\Theta'$ for some $\Lambda'$ and $\Theta'$
    \item $\Theta$ poised
    \item some $\Psi \in \Theta$ is stuck
  \end{enumerate}
\end{proof}

\begin{rem}
Another paper~\cite{Balzer19esop} introduces additional static restrictions to allow a stronger and more common notion of progress, which are orthogonal to our results. We expect that adopting this extension to our work would give the usual notion of progress with deadlock freedom.
\end{rem}

\section{Related Work}
\label{sec:related}
Our paper serves as an extension to the manifest sharing system defined in \cite{Balzer17icfp} by introducing a notion of subtyping to the
system which allows us to statically relax the equi-synchronizing constraint. Early glimpses of subtyping can be seen in the previous system
with the introduction of $\bot$ and $\top$ as the minimal and maximal constraints, which happened to be compatible with our subtyping
relation.

Subtyping for session types was first proposed by Gay and Hole~\cite{Gay05acta}, which was done in the classical setting for the linear
connectives except for $\Lupll$ and $\Ldownll$. 
Subtyping for the intuitionistic setting that we work on was also formalized by \cite{Acay16itrs}, which worked out subtyping for the linear
connectives except for $\Lupll$ and $\Ldownll$. That paper also introduces subtyping for intersection and union types, which are orthogonal
and thus compatible to the subtyping in our system. Neither of these papers investigates modalities or sharing, which are two of
our contributions to the understanding of subtyping. We believe that with a well-defined translation of modal shifts and the sharing semantics to the
classical setting, the subtyping on the shifts could be defined in the classical setting as well.

There have also been many recent developments in subtyping in the context of multiparty session types
\cite{Chen14ppdp,Chen17lmcs,Ghilezan19lamp,Ghilezan21popl}, which are a different class of type systems that describe protocols between an
arbitrary number of participants from a neutral global point of view. These systems are quite different in how they interpret subtyping,
since the subtyping we work with are at the channel level, where two communicating processes can safely disagree on the protocol. This
creates a fairly simple definition where subtyping is tightly coupled with the individual connectives.
However, since global types in multiparty session types can be projected to a binary setting, there may be non-obvious connections that
could be drawn. Thus, understanding the relation of our subtyping system to these systems is a challenge and an interesting item for future work.

\section{Conclusion}
\label{sec:conclusion}
We propose a subtyping extension to a message passing concurrency programming language introduced in previous work~\cite{Balzer17icfp} and showed examples highlighting the expressiveness that this new system provides. Throughout the paper, we follow two important principles, \emph{substitutability} and \emph{ignorance is bliss}, which gave a rich type system that in particular allows \emph{phases} (in a shared setting) to be manifest in the type.

One immediate application of shared subtyping is that combined with refinement types~\cite{Das20concur,Das21csf}, it can encode finer specifications of protocols. For example in the auction scenario, we can statically show that each client that does not win a bid gets refunded precisely the exact amount of money it bid. Without shared to linear subtyping, specifications of shared communication across multiple acquire-release cycles were not possible.

A future work in a more theoretical platform is to extend the setting to adjoint logic~\cite{Pruiksma19places}, which provides a more general framework of reasoning about modal shifts in a message passing system. In particular, we found that affine session types, where contraction (aliasing) is rejected, have immediate applications.

\paragraph{Acknowledgements}
We would like to thank the anonymous reviewers for feedback on the initially submitted version of this paper in COORDINATION 2021.
Supported by NSF Grant No. CCF-1718267 ``Enriching Session Types for Practical Concurrent Programming''.

\bibliographystyle{alpha}
\bibliography{fp,cites}

\newcommand{\etalchar}[1]{$^{#1}$}
\begin{thebibliography}{CDCSY17}

\bibitem[AP16]{Acay16itrs}
Co{\c{s}}ku Acay and Frank Pfenning.
\newblock Intersections and unions of session types.
\newblock In N.~Kobayashi, editor, {\em 8th Workshop on Intersection Types and
  Related Systems (ITRS'16)}, pages 4--19, Porto, Portugal, June 2016. EPTCS
  242.

\bibitem[BP17]{Balzer17icfp}
Stephanie Balzer and Frank Pfenning.
\newblock Manifest sharing with session types.
\newblock In {\em International Conference on Functional Programming (ICFP)},
  pages 37:1--37:29. ACM, September 2017.
\newblock Extended version available as Technical Report
  \href{http://www.cs.cmu.edu/~fp/papers/CMU-CS-17-106R.pdf}{CMU-CS-17-106R},
  June 2017.

\bibitem[BTP19]{Balzer19esop}
Stephanie Balzer, Bernardo Toninho, and Frank Pfenning.
\newblock Manifest deadlock-freedom for shared session types.
\newblock In L.~Caires, editor, {\em 28th European Symposium on Programming
  (ESOP 2019)}, pages 611--639, Prague, Czech Republic, April 2019. Springer
  LNCS 11423.

\bibitem[CDCSY17]{Chen17lmcs}
Tzu-chun Chen, Mariangiola Dezani-Ciancaglini, Alceste Scalas, and Nobuko
  Yoshida.
\newblock {On the Preciseness of Subtyping in Session Types}.
\newblock {\em {Logical Methods in Computer Science}}, {Volume 13, Issue 2},
  June 2017.

\bibitem[CDCY14]{Chen14ppdp}
Tzu-Chun Chen, Mariangiola Dezani-Ciancaglini, and Nobuko Yoshida.
\newblock On the preciseness of subtyping in session types.
\newblock In {\em Proceedings of the Conference on Principles and Practice of
  Declarative Programming (PPDP'14)}, Canterbury, UK, September 2014. ACM.

\bibitem[CHP99]{Crary99pldi}
Karl Crary, Robert Harper, and Sidd Puri.
\newblock What is a recursive module?
\newblock In {\em In SIGPLAN Conference on Programming Language Design and
  Implementation}, pages 50--63. ACM Press, 1999.

\bibitem[CP10]{Caires10concur}
Lu{\'\i}s Caires and Frank Pfenning.
\newblock Session types as intuitionistic linear propositions.
\newblock In {\em Proceedings of the 21st International Conference on
  Concurrency Theory (CONCUR 2010)}, pages 222--236, Paris, France, August
  2010. Springer LNCS 6269.

\bibitem[CS09]{Cervesato09ic}
Iliano Cervesato and Andre Scedrov.
\newblock Relating state-based and process-based concurrency through linear
  logic.
\newblock {\em Information and Computation}, 207(10):1044--1077, October 2009.

\bibitem[DBH{\etalchar{+}}21]{Das21csf}
Ankush Das, Stephanie Balzer, Jan Hoffmann, Frank Pfenning, and Ishani
  Santurkar.
\newblock Resource-aware session types for digital contracts.
\newblock In R.~K{\"u}sters and D.~Naumann, editors, {\em 34th Computer
  Security Foundations Symposium (CSF 2021)}, Dubrovnik, Croatia, June 2021.
  IEEE.
\newblock To appear.

\bibitem[DP20]{Das20concur}
Ankush Das and Frank Pfenning.
\newblock Session types with arithmetic refinements.
\newblock In I.~Konnov and L.~Kov{\'a}cs, editors, {\em 31st International
  Conference on Concurrency Theory (CONCUR 2020)}, pages 13:1--13:18, Vienna,
  Austria, September 2020. LIPIcs 171.

\bibitem[GH05]{Gay05acta}
Simon~J. Gay and Malcolm Hole.
\newblock Subtyping for session types in the {$\pi$}-calculus.
\newblock {\em Acta Informatica}, 42(2--3):191--225, 2005.

\bibitem[GJP{\etalchar{+}}19]{Ghilezan19lamp}
Silvia Ghilezan, Svetlana Jakšić, Jovanka Pantović, Alceste Scalas, and
  Nobuko Yoshida.
\newblock Precise subtyping for synchronous multiparty sessions.
\newblock {\em Journal of Logical and Algebraic Methods in Programming},
  104:127 -- 173, 2019.

\bibitem[GPP{\etalchar{+}}20]{Ghilezan21popl}
Silvia Ghilezan, Jovanka Pantović, Ivan Prokić, Alceste Scalas, and Nobuko
  Yoshida.
\newblock Precise subtyping for asynchronous multiparty sessions, 2020.

\bibitem[Gri15]{Griffith15phd}
Dennis Griffith.
\newblock {\em Polarized Substructural Session Types}.
\newblock PhD thesis, University of Illinois at Urbana-Champaign, 2015.
\newblock In preparation.

\bibitem[Hon93]{Honda93concur}
Kohei Honda.
\newblock Types for dyadic interaction.
\newblock In E.~Best, editor, {\em 4th International Conference on Concurrency
  Theory (CONCUR 1993)}, pages 509--523. Springer LNCS 715, 1993.

\bibitem[HVK98]{Honda98esop}
Kohei Honda, Vasco~T. Vasconcelos, and Makoto Kubo.
\newblock Language primitives and type discipline for structured
  communication-based programming.
\newblock In C.~Hankin, editor, {\em 7th European Symposium on Programming
  Languages and Systems (ESOP 1998)}, pages 122--138. Springer LNCS 1381, 1998.

\bibitem[MM84]{Mitchell84podc}
Don~P. Mitchell and Michael Merritt.
\newblock A distributed algorithm for deadlock detection and resolution.
\newblock In {\em Symposium on Principles of Distributed Computation (PODC
  1984)}, pages 282--284, Vancouver, British Columbia, August 1984. ACM.

\bibitem[PG15]{Pfenning15fossacs}
Frank Pfenning and Dennis Griffith.
\newblock Polarized substructural session types.
\newblock In A.~Pitts, editor, {\em Proceedings of the 18th International
  Conference on Foundations of Software Science and Computation Structures
  (FoSSaCS 2015)}, pages 3--22, London, England, April 2015. Springer LNCS
  9034.
\newblock Invited talk.

\bibitem[PP19]{Pruiksma19places}
Klaas Pruiksma and Frank Pfenning.
\newblock A message-passing interpretation of adjoint logic.
\newblock In F.~Martins and D.~Orchard, editors, {\em Workshop on Programming
  Language Approaches to Concurrency and Communication-Centric Software
  (PLACES)}, pages 60--79, Prague, Czech Republic, April 2019. EPTCS 291.

\bibitem[San19]{Sano19ms}
Chuta Sano.
\newblock On session typed contracts for imperative languages.
\newblock Masters thesis, Carnegie Mellon University, December 2019.
\newblock Available as Technical Report
  \href{http://reports-archive.adm.cs.cmu.edu/anon/2019/CMU-CS-19-133.pdf}{CMU-CS-19-133},
  December 2019.

\bibitem[SBP21]{Sano21coord}
Chuta Sano, Stephanie Balzer, and Frank Pfenning.
\newblock Manifestly phased communication via shared session types.
\newblock In Ferruccio Damiani and Ornela Dardha, editors, {\em Coordination
  Models and Languages}, pages 23--40, Valletta, Malta, 2021. Springer LNCS
  12717.

\bibitem[Ton15]{Toninho15phd}
Bernardo Toninho.
\newblock {\em A Logical Foundation for Session-based Concurrent Computation}.
\newblock PhD thesis, Carnegie Mellon University and Universidade Nova de
  Lisboa, May 2015.
\newblock Available as Technical Report CMU-CS-15-109.

\bibitem[Wad12]{Wadler12icfp}
Philip Wadler.
\newblock Propositions as sessions.
\newblock In {\em Proceedings of the 17th International Conference on
  Functional Programming (ICFP 2012)}, pages 273--286, Copenhagen, Denmark,
  September 2012. ACM Press.

\end{thebibliography}

\clearpage

\appendix
\section{Meet Operator}
$\hat{A} \land \hat{B}$ is defined coinductively from the structure of its arguments. Note that there are many cases where these rules do not
apply -- in that case the result of the meet is $\bot$.

\begin{small}
  \begin{gather*}
    1 \land 1 \to 1 \\
    \tensor{A}{A'} \land \tensor{B}{B'} \to (\Lo{A} \land \Lo{B}) \otimes (\Lo{A'} \land \Lo{B'}) \\
    \loli{A}{A'} \land \loli{B}{B'} \to (\Lo{A} \land \Lo{B}) \multimap (\Lo{A'} \land \Lo{B'}) \\
    \extch{\caselist{l}{A}, \caselist{m}{B}} \land \extch{\caselist{l}{A'}, \caselist{n}{C}} \to
    \extch{\overline{l:(\Lo{A} \land \Lo{A'})}, \caselist{m}{B}, \caselist{n}{C}} \\
    \intch{\caselist{l}{A}, \caselist{m}{B}} \land \intch{\caselist{l}{A'}, \caselist{n}{C}}
    \to \intch{\overline{l:(\Lo{A} \land \Lo{A'})}} \tag{$\overline{l}$ not empty} \\
    \upls{A} \land \upls{B} \to \Lupls{(\Lo{A} \land \Lo{B})} \\
    \upls{A} \land \upll{B} \to \Lupls{(\Lo{A} \land \So{B})} \\
    \upll{A} \land \upls{B} \to \Lupls{(\So{A} \land \Lo{B})} \\
    \upll{A} \land \upll{B} \to \Lupll{(\So{A} \land \So{B})} \\
    \downsl{A} \land \downsl{B} \to \Ldownsl{(\So{A} \land \So{B})} \\
    \downsl{A} \land \downll{B} \to \Ldownsl{(\So{A} \land \Lo{B})} \\
    \downll{A} \land \downsl{B} \to \Ldownsl{(\Lo{A} \land \So{B})} \\
    \downll{A} \land \downll{B} \to \Ldownll{(\Lo{A} \land \Lo{B})}
  \end{gather*}
\end{small}%

Intuitively, the idea with this construction is that on external choices, we take the union of the labels on both sides whereas on internal choices, we take the
intersection of the labels on both sides. Since we do not allow the nullary internal choice $\intch{}$ in the language, we require that the meet between two
internal choices to be non-empty, that is, they must share at least one label. Otherwise, the meet construction should produce a $\bot$.

\begin{lem}
  $\hat{A} \land \hat{B}$ is the greatest lower bound between $\hat{A}$ and $\hat{B}$ with respect to subtyping.
  \label{lem:subtype-meet}
\end{lem}
\begin{proof}
  By coinduction on the construction rules. The interesting part is on the external and internal choices; the construction tightly matches
  the appropriate direction of subtyping in the sense that the set of labels grows on external choices and shrinks on internal choices.
\end{proof}

\section{Proofs of Lemmas}
\label{app:lemma-proofs}
\begin{lem}
  If $\JCl{\Gamma}{\Psi, \Theta}{\Delta}$, then $\JCl{\Gamma}{\Theta}{\Delta'}$ for some $\Delta'$. \\
  More generally, if $\JCl{\Gamma}{\Theta_1, \Theta_2}{\Delta}$, then $\JCl{\Gamma}{\Theta_2}{\Delta'}$ for some $\Delta'$.
\end{lem}
\begin{proof}
  For the first part, by case analysis on the derivation of $\JCl{\Gamma}{\Psi, \Theta}{\Delta}$.
  In both cases ($\Theta2$ and $\Theta3$), we directly see that $\JCl{\Gamma}{\Theta}{\Delta'}$ for some $\Delta'$. \\
  For the second part, we can repeatedly apply the first part sequentially for every $\Psi \in \Theta_1$.
\end{proof}

\begin{lem}
  Given a well-formed $\Lambda;\Theta$, $\forall \procs{a}{-} \in \Lambda, \Psi_a \notin \Theta$
\end{lem}
\begin{proof}
  By well-formedness of $\Lambda$, $\procs{a}{-} \in \Lambda$ means that $\unavail{a} \notin
  \Lambda$. By the contrapositive of well-formedness of $\Lambda;\Theta$, $\unavail{a} \notin
  \Lambda \implies \Psi_a \notin \Theta$
\end{proof}

\begin{lem}
  If $\JCl{\Gamma}{\Psi_a, \Theta_1, \Psi_b, \Theta_2}{\Lt{a}{A}, \Delta}$ and $\Psi_a$ uses
  $\Lo{b}$, then
  $$\JCl{\Gamma}{\Psi_a, \Psi_b, \Theta_1, \Theta_2}{\Lt{a}{A}, \Delta}$$
\end{lem}
\begin{proof}
  By well-formedness, $\Psi_b$ is the only process in the configuration offering $\Lo{b}$.
  Furthermore by linearity, there can only be one process that use $\Lo{b}$, which is $\Psi_a$ by
  assumption, so $\Lo{b}$ will not be consumed by any processes in $\Theta_1$. Therefore, we can
  repeatedly move $\Psi_b$ to the left in the configuration until it is to the right of $\Psi_a$,
  the unique process using $\Lo{b}$.
\end{proof}

\begin{lem}
  If $\JCl{\Gamma}{\Psi, \Theta}{\Delta}$, $\JCl{\Gamma}{\Theta}{\Delta_p}$, and $\JCl{\Gamma}{\Theta'}{\Delta_p}$, then
  $$\JCl{\Gamma}{\Psi, \Theta'}{\Delta}$$  \\
  More generally, if $\JCl{\Gamma}{\Theta_1, \Theta_2}{\Delta}$, $\JCl{\Gamma}{\Theta_2}{\Delta_p}$, and $\JCl{\Gamma}{\Theta_2'}{\Delta_p}$, then
  $$\JCl{\Gamma}{\Theta_1, \Theta_2}{\Lt{a}{A}, \Delta}$$
\end{lem}
\begin{proof}
  For the first part, by case analysis on the derivation of $\JCl{\Gamma}{\Psi, \Theta}{\Delta}$.
  In both cases ($\Theta2$ and $\Theta3$), we can directly substitute $\Theta'$ for $\Theta$ where it appears in the configuration judgment. \\
  For the second part, we can repeatedly apply the first part sequentially for every $\Psi \in \Theta_1$.
\end{proof}

\begin{lem}
  If $\JCl{\Gamma}{\Psi_a, \Theta'}{\Lt{a}{A'}, \Delta'}$, then for any $\Lo{B}$ such that $\Lo{A'} \leq \Lo{B}$,\\
  $\JCl{\Gamma}{\Psi_a, \Theta'}{\Lt{a}{B}, \Delta'}$.
\end{lem}
\begin{proof}
  By inversion on the derivation of $\JCl{\Gamma}{\Psi_a, \Theta'}{\Lt{a}{A'}, \Delta}$.
  \begin{case}
    $$\infer[\Theta 2]{\JCpc{\Gamma}{a}{b}{\Lo{A}}{\Theta'}{\Lt{a}{A'}, \Delta'}}{\Sto{b}{\hat{B}} \in \Gamma & \hat{b} \leq \Lo{A} &\JCl{\Gamma}{\Theta'}{\Delta'}}$$
    \noindent By transitivity, $\hat{B} \leq \Lo{B}$ therefore $\JCpc{\Gamma}{a}{b}{\Lo{A}}{\Theta'}{\Lt{a}{B}, \Delta'}$
  \end{case}
  \begin{case}
    $$\infer[\Theta 3]{\JCpp{\Gamma}{a}{P}{\Lo{A}}{\Theta'}{\Delta'}}
    {\Sto{a}{\hat{A}} \in \Gamma & \dsyncl{A'}{A}{\hat{A}} & \JTl{\Gamma}{\Delta_a}{P}{a}{A'} & \JCl{\Gamma}{\Theta'}{\Delta_a, \Delta'}}$$
    \noindent By transitivity, $\Lo{A'} \leq \Lo{B}$ and therefore $\dsyncl{A'}{B}{\hat{A}}$ by Lemma~\ref{lem:dsync-bigger}.
    \\Therefore, $\JCpp{\Gamma}{a}{P}{\Lo{A}}{\Theta'}{\Delta'}$
  \end{case}
\end{proof}

\begin{lem}
  Let $\Gamma' \preceq \Gamma$ and $\JTl{\Gamma}{\Delta}{P}{z}{C}$, then $\JTl{\Gamma'}{\Delta}{P}{z}{C}$.
\end{lem}
\begin{proof}
  We first prove the admissibility of the substitution of a shared channel by a smaller type in a typing judgment. In particular, we will
  begin by showing that if
  $$\JTl{\Gamma, \Sto{x}{\hat{A}}}{\Delta}{P}{z}{C}$$
  then $$\JTl{\Gamma, \Sto{x}{\hat{B}}}{\Delta}{P}{z}{C}$$
  for some $\hat{B}\leq \hat{A}$ by induction on the derivation of $\JTl{\Gamma, \Sto{x}{\hat{A}}}{\Delta}{P}{z}{C}$.

  First, we begin by pointing out that rules that do not use $\So{x}$ (most of them) are trivial since we can just appeal to the induction
  hypothesis (IH) on the premise(s) in the appropriate derivation. The rules that can use $\So{x}$ are $\So{ID}, ID_{\scriptscriptstyle
  {LS}}, SP_{\scriptscriptstyle {LL}}, SP_{\scriptscriptstyle {LS}}, SP_{\scriptscriptstyle {SS}}, \Lupls L, {\multimap} \So{L},$ and
  ${\otimes} \So{R}$. For these cases, we can confirm that the substitution is valid by using the IH and using transitivity of $\leq$. We will present one such case:
  \begin{case}
    $$\infer[{\otimes} \So{R}]{\JTlr{\Gamma, \Sto{x}{\hat{A}}}{\Delta}{\sends{y}{x}; P}{y}{\tensor{A}{B}}}{\hat{A} \leq \Lo{A} & \JTl{\Gamma, \Sto{x}{\hat{A}}}{\Delta}{P}{y}{B}}$$
    Then by IH, $\JTl{\Gamma, \Sto{x}{\hat{B}}}{\Delta}{P}{y}{B}$. Furthermore, by transitivity, $\hat{B} \leq \Lo{A}$. Therefore by ${\otimes} \So{R}$, $\JTlr{\Gamma, \Sto{x}{\hat{A}}}{\Delta}{\sends{y}{x}; P}{y}{\tensor{A}{B}}$
  \end{case}

  After showing that substitution by a smaller type in the shared context $\Gamma$ is admissible, the remaining part is to note that
  $\Gamma'$ either contains additional channels that is in $\Gamma,$ which we repeat the argument above for, or $\Gamma'$
  contains new channel names compared to $\Gamma$, which we resolve via weakening.
\end{proof}
\begin{lem}
  Let $\Gamma' \preceq \Gamma$ then
  \begin{enumerate}
    \item If $\JCl{\Gamma}{\Theta}{\Delta}$ for some $\Theta, \Delta$, then $\JCl{\Gamma'}{\Theta}{\Delta}$
    \item If $\JCs{\Gamma}{\Lambda}{\Gamma''}$ for some $\Lambda, \Gamma''$, then $\JCs{\Gamma'}{\Theta}{\Gamma''}$
  \end{enumerate}
\end{lem}
\begin{proof}
  For the first part, by induction on the derivation of $\JCl{\Gamma}{\Theta}{\Delta}$.
  \begin{case}
    $$\infer[\Theta 1]{\JCl{\Gamma}{\cdot}{\cdot}}{}$$
    Any $\Gamma$ applies, so in particular any $\Gamma' \preceq \Gamma$ will as well. 
  \end{case}
  \begin{case}
    $$\infer[\Theta 2]{\JCpc{\Gamma}{a}{b}{\Lo{A}}{\Theta'}{\Delta'}}{\Sto{b}{\hat{B}} \in \Gamma & \hat{B} \leq \Lo{A} &\JCl{\Gamma}{\Theta'}{\Delta'}}$$
    By exchange, we can assume without loss of generality that $\Gamma = \Sto{b}{\hat{B}}, \Gamma_r$. Similarly, we can assume without loss of generality that $\Gamma' = \Sto{b}{\hat{B'}}, \Gamma_r'$ where $\hat{B'} \leq \hat{B}$ and $\Gamma_r' \preceq \Gamma$. \\
    $\hat{B'} \leq \Lo{A}$ follows by transitivity of $\leq$ and $\JCl{\Gamma'}{\Theta'}{\Delta_a, \Delta'}$ follows from the IH. Therefore,
    $$\JCpc{\Gamma'}{a}{b}{\Lo{A}}{\Theta'}{\Delta'}$$
  \end{case}
  \begin{case}
    $$\infer[\Theta 3]{\JCpp{\Gamma}{a}{P}{\Lo{A}}{\Theta'}{\Delta'}}
    {\Sto{a}{\hat{A}} \in \Gamma & \dsyncl{A'}{A}{\hat{A}} & \JTl{\Gamma}{\Delta_a}{P}{a}{A'}	& \JCl{\Gamma}{\Theta'}{\Delta_a, \Delta'}}$$
    By exchange, we can assume without loss of generality that $\Gamma = \Sto{a}{\hat{A}}, \Gamma_r$. Similarly, we can assume without loss of generality that $\Gamma' = \Sto{a}{\hat{A'}}, \Gamma_r'$ where $\hat{A'} \leq \hat{A}$ and $\Gamma_r' \preceq \Gamma$. \\
    $\dsyncl{A'}{A}{\hat{A'}}$ follows from Lemma~\ref{lem:dsync-smaller-hat}, $\JTl{\Gamma'}{\Delta_a}{P}{a}{A'}$ follows from
    Lemma~\ref{lem:static-stable}, and \\ $\JCl{\Gamma'}{\Theta'}{\Delta_a, \Delta'}$ follows from the IH. Therefore,
    $$\JCpp{\Gamma'}{a}{P}{\Lo{A}}{\Theta'}{\Delta'}$$
  \end{case}

 	\setcounter{case}{0}
  For the second part, by induction on the derivation of $\JCs{\Gamma}{\Lambda}{\Delta'}$
  \begin{case}
  	$$\infer[\Lambda 1]{\JCs{\Gamma}{\cdot}{\cdot}}{}$$
		Any $\Gamma$ applies, so in particular any $\Gamma' \preceq \Gamma$ will as well. 
  \end{case}
  \begin{case}
  	$$\infer[\Lambda 2]{\JCs{\Gamma}{\Lambda_1, \Lambda_2}{\Gamma_1, \Gamma_2}}{\JCs{\Gamma}{\Lambda_1}{\Gamma_1} & \JCs{\Gamma}{\Lambda_2}{\Gamma_2}}$$
  	Both $\JCs{\Gamma'}{\Lambda_1}{\Gamma_1}$ and $\JCs{\Gamma'}{\Lambda_2}{\Gamma_2}$ follow from the IH. Therefore,
  	$$\JCs{\Gamma'}{\Lambda_1, \Lambda_2}{\Gamma_1, \Gamma_2}$$
  \end{case}
  \begin{case}
  	$$\infer[\Lambda 3]{\JCs{\Gamma}{\procs{a}{P}}{\St{a}{A}}}{\dsyncs{A'}{A}{\top} & \JTs{\Gamma}{P}{a}{A'}}$$
    $\JTs{\Gamma'}{P}{a}{A'}$ follows from Lemma~\ref{lem:static-stable}. Therefore,
    $$\JCs{\Gamma}{\procs{a}{P}}{\St{a}{A}}$$
  \end{case}
  \begin{case}
  	$$\infer[\Lambda 4]{\JCs{\Gamma}{\unavail{a}}{\Sto{a}{\hat{A}}}}{}$$
		Any $\Gamma$ applies, so in particular any $\Gamma' \preceq \Gamma$ will as well. 
  \end{case}
\end{proof}

To prove the following lemmas, we switch to a set-based formulation of safe synchronization; $\dsync{A}{B}{\hat{D}}$ is written as $(A, B,
\hat{D}) \in \dstxt$. We also define a monotone map $F$ from the coinductive definition of ssync, giving us $\dstxt \in F(\dstxt)$; that
is, ssync is $F$-consistent. 
\begin{lem}
  If $A \leq B \leq C$ with all same modalities (that is, $A, B, C$ are either all linear or all shared) and $\dsync{A}{B}{\hat{D}}$, then
$\dsync{A}{C}{\hat{D}}$ for some $\hat{D}$.
\end{lem}
\begin{proof}
  We want to show that
  $$\dstxt' \defeq \dstxt \cup \dstxt_{\Uparrow}$$
  is $F$-consistent where
  \begin{align*}
    \dstxt_{\Uparrow} \defeq \{(A, C, \hat{D})
      \defor \exists B . B \leq C \land (A, B, \hat{D})
    \in \dstxt\}
  \end{align*}
  Again, where $A, B, C$ must all be of the same modality.\\
  We will prove $F$-consistency of $\dstxt'$, that is, $\dstxt' \in F(\dstxt')$ by
  showing that each of the two sets $\dstxt$ and $\dstxt_{\Uparrow}$ are subsets
  of $F(\dstxt')$. \\
  First, $\dstxt \subseteq F(\dstxt')$ immediately follows because $\dstxt \subseteq F(\dstxt)$ and $F(\dstxt) \subseteq F(\dstxt')$ by
  monotonicity of $F$ given $\dstxt \subseteq \dstxt'$.
  We will now consider $\dstxt_{\Uparrow} \in F(\dstxt')$ by case analysis on the structure of $A$. We can uniquely infer the structure of
  $B$ and $C$ from the structure of $A$ by inversion on the appropriate subtyping rule for most cases.
\begin{case}
  $A = \upll{A'}$; then $B = \upll{B'}$ and $C = \upll{C'}$ with $\Lo{A'} \leq \Lo{B'} \leq \Lo{C'}$.
  \begin{flalign*}
    (\upll{A'}, \upll{C'}, \hat{D}) &\in \dstxt_{\Uparrow}, (\upll{A'}, \upll{B'}, \hat{D}) \in \dstxt \tag{this case} \\
    (\Lo{A'}, \Lo{B'}, \hat{D}) &\in \dstxt \tag{by inversion on $D\Lupll$} \\
    (\Lo{A'}, \Lo{C'}, \hat{D}) &\in \dstxt_{\Uparrow} \tag{by definition of $\dstxt_{\Uparrow}$ with $\Lo{B'} \leq \Lo{C'}$} \\
    (\Lo{A'}, \Lo{C'}, \hat{D}) &\in \dstxt' \tag{since $\dstxt_{\Uparrow} \subseteq \dstxt'$} \\
    (\upll{A'}, \upll{C'}, \hat{D}) &\in F(\dstxt') \tag{by $D\Lupll$}
  \end{flalign*}
\end{case}
$\Ldownll, \otimes,$ and $\multimap$ follow a similar pattern of appealing to
the covariance of subtyping on the continuation types. 
\begin{case}
  $A = \intch{\caselist{l}{A}}$; then
  $B = \intch{\caselist{l}{B}, \caselist{m}{B}}$ and
  $C = \intch{\caselist{l}{C}, \caselist{m}{C}, \caselist{n}{C}}$ with \\
  $\Lo{{A_i}} \leq \Lo{{B_i}} \leq \Lo{{C_i}} \; \forall i \in \overline{l}$ and $\Lo{{B_i}} \leq \Lo{{C_i}} \; \forall i \in \overline{m}$.
  \begin{flalign*}
    (\intch{\caselist{l}{A}}, \intch{\caselist{l}{C}, \caselist{m}{C},
    \caselist{n}{C}}, \hat{D}) &\in \dstxt_{\Uparrow} \\
    (\intch{\caselist{l}{A}}, \intch{\caselist{l}{B}, \caselist{m}{B}},
    \hat{D}) &\in \dstxt
    \tag{this case} \\
    (\forall i \in \overline{l}) \;
    (\Lo{{A_i}}, \Lo{{B_i}}, \hat{D}) &\in \dstxt
    \tag{by inversion on $D{\oplus}$} \\
    (\forall i \in \overline{l}) \;
    (\Lo{{A_i}}, \Lo{{C_i}}, \hat{D}) &\in \dstxt_{\Uparrow} \\
    \tag{by definition of $\dstxt_{\Uparrow}$ with
    $\Lo{{B_i}} \leq \Lo{{C_i}}$} \\
    (\forall i \in \overline{l}) \;
    (\Lo{{A_i}}, \Lo{{C_i}}, \hat{D}) &\in \dstxt'
    \tag{since $\dstxt_{\Uparrow} \subseteq \dstxt'$} \\
    (\intch{\caselist{l}{A}}, \intch{\caselist{l}{C}, \caselist{m}{C},
    \caselist{n}{C}}, \hat{D}) &\in F(\dstxt')
    \tag{by $D{\oplus}$}
  \end{flalign*}
\end{case}
$D{\&}$ follows a similar pattern.

\begin{case}
  $A = \downsl{A}$; then there are three possible assignments to $B$ and $C$ that satisfies the subtyping constraints, so we will continue by subcasing on the structure of $B$ and $C$.
  \begin{subcase}
    $B = \downsl{B}$ and $C = \downsl{C}$ with $\So{A} \leq \So{B} \leq \So{C}$.
    \begin{flalign*}
      (\downsl{A}, \downsl{C}, \hat{D}) &\in \dstxt_{\Uparrow}, (\downsl{A}, \downsl{B}, \hat{D}) \in \dstxt \tag{this case} \\
      (\So{A}, \So{B}, \top) &\in \dstxt, \So{A} \leq \hat{D} \tag{by inversion on $D\Ldownsl$} \\
      (\So{A}, \So{C}, \top) &\in \dstxt_{\Uparrow} \tag{by definition of $\dstxt_{\Uparrow}$ with $\So{B} \leq \So{C}$} \\
      (\So{A}, \So{C}, \top) &\in \dstxt' \tag{since $\dstxt_{\Uparrow} \subseteq \dstxt'$} \\
      (\downsl{A}, \downsl{C}, \hat{D}) &\in F(\dstxt') \tag{by $D\Ldownsl$}
    \end{flalign*}
  \end{subcase}
  \begin{subcase}
    $B = \downsl{B}$ and $C = \downll{C}$ with $\So{A} \leq \So{B} \leq \Lo{C}$.
    \begin{flalign*}
      (\downsl{A}, \downll{C}, \hat{D}) &\in \dstxt_{\Uparrow}, (\downsl{A}, \downsl{B}, \hat{D}) \in \dstxt \tag{this case} \\
      (\So{A}, \So{B}, \top) &\in \dstxt, \So{A} \leq \hat{D} \tag{by inversion on $D\Ldownsl$} \\
      (\So{A}, \Lo{C}, \top) &\in \dstxt_{\Uparrow} \tag{by definition of $\dstxt_{\Uparrow}$ with $\So{B} \leq \Lo{C}$} \\
      (\So{A}, \Lo{C}, \top) &\in \dstxt' \tag{since $\dstxt_{\Uparrow} \subseteq \dstxt'$} \\
      (\downsl{A}, \downll{C}, \hat{D}) &\in F(\dstxt') \tag{by $D\Ldownsl$}
    \end{flalign*}
  \end{subcase}
  \begin{subcase}
    $B = \downll{B}$ and $C = \downll{C}$ with $\So{A} \leq \Lo{B} \leq \Lo{C}$.
    \begin{flalign*}
      (\downsl{A}, \downll{C}, \hat{D}) &\in \dstxt_{\Uparrow}, (\downsl{A}, \downll{B}, \hat{D}) \in \dstxt \tag{this case} \\
      (\So{A}, \Lo{B}, \top) &\in \dstxt, \So{A} \leq \hat{D} \tag{by inversion on $D\Ldownsl\Ldownll$} \\
      (\So{A}, \Lo{C}, \top) &\in \dstxt_{\Uparrow} \tag{by definition of $\dstxt_{\Uparrow}$ with $\Lo{B} \leq \Lo{C}$} \\
      (\So{A}, \Lo{C}, \top) &\in \dstxt' \tag{since $\dstxt_{\Uparrow} \subseteq \dstxt'$} \\
      (\downsl{A}, \downll{C}, \hat{D}) &\in F(\dstxt') \tag{by $D\Ldownsl$}
    \end{flalign*}
  \end{subcase}
\end{case}

\begin{case}
  $A = \upls{A}$; then there are three possible assignments to $B$ and $C$ that satisfies the subtyping constraints, so we will continue by subcasing on the structure of $B$ and $C$. 
  \begin{subcase}
    $B = \upls{B}$ and $C = \upls{C}$	with $\Lo{A} \leq \Lo{B} \leq \Lo{C}$.
    \begin{flalign*}
      (\upls{A}, \upls{C}, \top) &\in \dstxt_{\Uparrow}, (\upls{A}, \upls{B}, \top) \in \dstxt \tag{this case} \\
      (\Lo{A}, \Lo{B}, \upls{A}) &\in \dstxt \tag{by inversion on $D\Lupls$} \\
      (\Lo{A}, \Lo{C}, \upls{A}) &\in \dstxt_{\Uparrow} \tag{by definition of $\dstxt_{\Uparrow}$ with $\Lo{A} \leq \Lo{B}$} \\
      (\Lo{A}, \Lo{C}, \upls{A}) &\in \dstxt' \tag{since $\dstxt_{\Uparrow} \subseteq \dstxt'$} \\
      (\upls{A}, \upls{C}, \top) &\in F(\dstxt') \tag{by $D\Lupls$}
    \end{flalign*}
  \end{subcase}
  \begin{subcase}
    $B = \upls{B}$ and $C = \upll{C}$	with $\Lo{A} \leq \Lo{B} \leq \Lo{C}$.
    \begin{flalign*}
      (\upls{A}, \upll{C}, \top) &\in \dstxt_{\Uparrow}, (\upls{A}, \upls{B}, \top) \in \dstxt \tag{this case} \\
      (\Lo{A}, \Lo{B}, \upls{A}) &\in \dstxt \tag{by inversion on $D\Lupls$} \\
      (\Lo{A}, \Lo{C}, \upls{A}) &\in \dstxt_{\Uparrow} \tag{by definition of $\dstxt_{\Uparrow}$ with $\Lo{A} \leq \Lo{B}$} \\
      (\Lo{A}, \Lo{C}, \upls{A}) &\in \dstxt' \tag{since $\dstxt_{\Uparrow} \subseteq \dstxt'$} \\
      (\upls{A}, \upll{C}, \top) &\in F(\dstxt') \tag{by $D\Lupls$}
    \end{flalign*}
  \end{subcase}
  \begin{subcase}
    $B = \upll{B}$ and $C = \upll{C}$	with $\Lo{A} \leq \Lo{B} \leq \Lo{C}$.
    \begin{flalign*}
      (\upls{A}, \upll{C}, \top) &\in \dstxt_{\Uparrow}, (\upls{A}, \upll{B}, \top) \in \dstxt \tag{this case} \\
      (\Lo{A}, \Lo{B}, \upls{A}) &\in \dstxt \tag{by inversion on $D\Lupls\Lupll$} \\
      (\Lo{A}, \Lo{C}, \upls{A}) &\in \dstxt_{\Uparrow} \tag{by definition of $\dstxt_{\Uparrow}$ with $\Lo{A} \leq \Lo{B}$} \\
      (\Lo{A}, \Lo{C}, \upls{A}) &\in \dstxt' \tag{since $\dstxt_{\Uparrow} \subseteq \dstxt'$} \\
      (\upls{A}, \upll{C}, \top) &\in F(\dstxt') \tag{by $D\Lupls$}
    \end{flalign*}
  \end{subcase}
\end{case}
We missed one case, when $A = B = C = 1$, but this case is trivial since
$\dstxt_{\Uparrow}$ does not add any new members to the set.
\end{proof}

\begin{lem}
  If $A \leq B \leq C$ with all same modalities, $\dsync{B}{C}{\hat{D}}$, and \\$\dsync{A}{C}{\hat{E}}$, then $\dsync{A}{C}{\hat{D}}$ for some $\hat{D}$ and $\hat{E}$.
\end{lem}
\begin{proof}
  We want to show that
  $$\dstxt' \defeq \dstxt \cup \dstxt_{\Downarrow}$$
  is $F$-consistent with
  \begin{align*}
    \dstxt_{\Downarrow} &\defeq \{(A, C, \hat{D})
      \defor \exists B . A \leq B \land (B, C, \hat{D})
    \in \dstxt \land \exists \hat{E} . (A, C, \hat{E}) \in \dstxt\} \\
  \end{align*}
  The proof is very similar in style to the previous lemma, but there is one additional constraint that $(A, C, \hat{E}) \in \dstxt$ for any
  constraint $\hat{E}$. This assumption is only necessary for the $\Lupls$ case. \\ In any case, we will prove $F$-consistency of $\dstxt'$,
  that is, $\dstxt' \in F(\dstxt')$ by showing that each of the three sets $\dstxt$ and $\dstxt_{\Downarrow}$ are subsets of $F(\dstxt')$.
  \\
  First, $\dstxt \subseteq F(\dstxt')$ immediately follows from the same argument as in the previous proof.

  We will now consider $\dstxt_{\Downarrow} \in F(\dstxt')$ by case analysis on the structure of $B$. Because all of $A, B, C$ have the same
  modality, we can uniquely infer the structure of $A$ and $C$ from the structure of $B$ by inversion on the appropriate subtyping rule.
  \begin{case}
    $A = \upll{A'}$; then $B = \upll{B'}$ and $C = \upll{C'}$ with $\Lo{A'} \leq \Lo{B'} \leq \Lo{C'}$.
    \begin{flalign*}
      (\upll{A'}, \upll{C'}, \hat{D}) &\in \dstxt_{\Downarrow}, 
      (\upll{B'}, \upll{C'}, \hat{D}) \in \dstxt
      \tag{this case} \\
      (\Lo{B'}, \Lo{C'}, \hat{D}) &\in \dstxt
      \tag{by inversion on $D\Lupll$} \\
      (\Lo{A'}, \Lo{C'}, \hat{D}) &\in \dstxt_{\Downarrow}
      \tag{by definition of $\dstxt_{\Downarrow}$ with $\Lo{A'} \leq \Lo{B'}$} \\
      (\Lo{A'}, \Lo{C'}, \hat{D}) &\in \dstxt'
      \tag{since $\dstxt_{\Downarrow} \subseteq \dstxt'$} \\
      (\upll{A'}, \upll{C'}, \hat{D}) &\in F(\dstxt')
      \tag{by $D\Lupll$}
    \end{flalign*}
  \end{case}
  $\Ldownll, \otimes,$ and $\multimap$ follow a similar pattern of appealing to the covariance of subtyping on the continuation types. 
  \begin{case}
    $A = \intch{\caselist{l}{A}}$; then
    $B = \intch{\caselist{l}{B}, \caselist{m}{B}}$ and
    $C = \intch{\caselist{l}{C}, \caselist{m}{C}, \caselist{n}{C}}$ with \\
    $\Lo{{A_i}} \leq \Lo{{B_i}} \leq \Lo{{C_i}} \; \forall i \in \overline{l}$ and $\Lo{{B_i}} \leq \Lo{{C_i}} \; \forall i \in \overline{m}$.
    \begin{flalign*}
      (\intch{\caselist{l}{A}}, \intch{\caselist{l}{C}, \caselist{m}{C},
      \caselist{n}{C}}, \hat{D}) &\in \dstxt_{\Downarrow} \\
      (\intch{\caselist{l}{B}, \caselist{m}{B}}, \intch{\caselist{l}{C},
    \caselist{m}{C}, \caselist{n}{C}}, \hat{D}) &\in \dstxt
    \tag{this case} \\
    (\forall i \in \overline{l},\overline{m}) \;
    (\Lo{{B_i}}, \Lo{{C_i}}, \hat{D}) &\in \dstxt
    \tag{by inversion on $D{\oplus}$} \\
    (\forall i \in \overline{l}) \;
    (\Lo{{A_i}}, \Lo{{C_i}}, \hat{D}) &\in \dstxt_{\Downarrow} \\
    \tag{by definition of $\dstxt_{\Downarrow}$ with
    $\Lo{{A_i}} \leq \Lo{{B_i}}$} \\
    (\forall i \in \overline{l}) \;
    (\Lo{{A_i}}, \Lo{{C_i}}, \hat{D}) &\in \dstxt'
    \tag{since $\dstxt_{\Downarrow} \subseteq \dstxt'$} \\
    (\intch{\caselist{l}{A}}, \intch{\caselist{l}{C}, \caselist{m}{C},
  \caselist{n}{C}}, \hat{D}) &\in F(\dstxt')
  \tag{by $D{\oplus}$}
\end{flalign*}
\end{case}
$D{\&}$ follows a similar pattern.
\begin{case}
  $A = \downsl{A}$; similar to the proof of Lemma~\ref{lem:dsync-bigger}, there are three possible assignments for $B$ and $C$. We will present one of those subcases: let $B = \downsl{B}$ and $C = \downsl{C}$ with $\So{A} \leq \So{B} \leq \So{C}$. The other two cases are similar.
  \begin{flalign*}
    (\downsl{A}, \downsl{C}, \hat{D}) &\in \dstxt_{\Downarrow}, (\downsl{B}, \downsl{C}, \hat{D}) \in \dstxt \tag{this case} \\
    (\So{B}, \So{C}, \top) &\in \dstxt, \So{B} \leq \hat{D} \tag{by inversion on $D\Ldownsl$} \\
    (\So{A}, \So{C}, \top) &\in \dstxt_{\Downarrow} \tag{by definition of $\dstxt_{\Downarrow}$ with $\So{A} \leq \So{B}$} \\
    (\So{A}, \So{C}, \top) &\in \dstxt' \tag{since $\dstxt_{\Downarrow} \subseteq \dstxt'$} \\
    \So{A} \leq \hat{D}  \tag{because $\So{A} \leq \So{B} \leq \hat{D}$} \\
    (\downsl{A}, \downsl{C}, \hat{D}) &\in F{\dstxt'} \tag{by $D\Ldownsl$}
  \end{flalign*}
\end{case}

\begin{case}
  $A = \upls{A}$; again, there are three possible assignments for $B$ and $C$, and we will take the subcase when $B = \upls{B}$ and $C = \upls{C}$ with $\Lo{A} \leq \Lo{B} \leq \Lo{C}$. The other two cases are similar. This case finally uses our assumption that $(A, C, \hat{E}) \in \dstxt$ -- $\hat{E}$ must be $\top$ due to $A = \upls{A}$.
  \begin{flalign*}
    (\upls{A}, \upls{C}, \top) &\in \dstxt_{\Downarrow}, (\upls{B}, \upls{C}, \top) \in \dstxt \tag{this case} \\
    (\upls{A}, \upls{C}, \top) &\in \dstxt \tag{By assumption with $\hat{E} = \top$} \\
    (\Lo{A}, \Lo{C}, \upls{A}) &\in \dstxt \tag{by inversion on $D\Lupls$} \\
    (\Lo{A}, \Lo{C}, \upls{A}) &\in \dstxt' \tag{since $\dstxt \subseteq \dstxt'$} \\
    (\upls{A}, \upls{C}, \top) &\in F(\dstxt') \tag{by $D\Lupls$}
  \end{flalign*}
\end{case}
We missed one case, when $A = B = C = 1$, but this case is trivial since
$\dstxt_{\Downarrow}$ does not add any new members to the set.
\end{proof}

\begin{lem}
  If $\dsyncl{A}{B}{\hat{C}}$ and $\hat{D} \leq \hat{C}$, then $\dsyncl{A}{B}{\hat{D}}$ for some $\Lo{A}, \Lo{B}, \hat{C},$ and $\hat{D}$.
\end{lem}
\begin{proof}
  We want to show that
  $$\dstxt' \defeq \dstxt \cup \dstxt_{\Downarrow}$$
  is $F$-consistent with
  \begin{align*}
    \dstxt_{\Downarrow} &\defeq \{(A, C, \hat{D})
      \defor \exists B . A \leq B \land (B, C, \hat{D})
    \in \dstxt \land \exists \hat{E} . (A, C, \hat{E}) \in \dstxt\} \\
    \end{align*}
    The proof is very similar in style to the previous lemma, but there is one additional constraint that $(A, C, \hat{E}) \in \dstxt$ for
    any constraint $\hat{E}$. This assumption is only necessary for the $\Lupls$ case. \\ In any case, we will prove $F$-consistency of
    $\dstxt'$, that is, $\dstxt' \in F(\dstxt')$ by showing that each of the three sets $\dstxt$ and $\dstxt_{\Downarrow}$ are subsets of
    $F(\dstxt')$. \\
    First, $\dstxt \subseteq F(\dstxt')$ immediately follows from the same argument as in the previous proof.

    We will now consider $\dstxt_{\Downarrow} \in F(\dstxt')$ by case analysis on the structure of $B$. Because all of $A, B, C$ have the
    same modality, we can uniquely infer the structure of $A$ and $C$ from the structure of $B$ by inversion on the appropriate subtyping
    rule.
    \begin{case}
      $A = \upll{A'}$; then $B = \upll{B'}$ and $C = \upll{C'}$ with $\Lo{A'} \leq \Lo{B'} \leq \Lo{C'}$.
      \begin{flalign*}
        (\upll{A'}, \upll{C'}, \hat{D}) &\in \dstxt_{\Downarrow}, 
        (\upll{B'}, \upll{C'}, \hat{D}) \in \dstxt
        \tag{this case} \\
        (\Lo{B'}, \Lo{C'}, \hat{D}) &\in \dstxt
        \tag{by inversion on $D\Lupll$} \\
        (\Lo{A'}, \Lo{C'}, \hat{D}) &\in \dstxt_{\Downarrow}
        \tag{by definition of $\dstxt_{\Downarrow}$ with $\Lo{A'} \leq \Lo{B'}$} \\
        (\Lo{A'}, \Lo{C'}, \hat{D}) &\in \dstxt'
        \tag{since $\dstxt_{\Downarrow} \subseteq \dstxt'$} \\
        (\upll{A'}, \upll{C'}, \hat{D}) &\in F(\dstxt')
        \tag{by $D\Lupll$}
      \end{flalign*}
    \end{case}
    $\Ldownll, \otimes,$ and $\multimap$ follow a similar pattern of appealing to the covariance of subtyping on the continuation types. 
    \begin{case}
      $A = \intch{\caselist{l}{A}}$; then
      $B = \intch{\caselist{l}{B}, \caselist{m}{B}}$ and
      $C = \intch{\caselist{l}{C}, \caselist{m}{C}, \caselist{n}{C}}$ with \\
      $\Lo{{A_i}} \leq \Lo{{B_i}} \leq \Lo{{C_i}} \; \forall i \in \overline{l}$ and $\Lo{{B_i}} \leq \Lo{{C_i}} \; \forall i \in \overline{m}$.
      \begin{flalign*}
        (\intch{\caselist{l}{A}}, \intch{\caselist{l}{C}, \caselist{m}{C},
        \caselist{n}{C}}, \hat{D}) &\in \dstxt_{\Downarrow} \\
        (\intch{\caselist{l}{B}, \caselist{m}{B}}, \intch{\caselist{l}{C},
      \caselist{m}{C}, \caselist{n}{C}}, \hat{D}) &\in \dstxt
      \tag{this case} \\
      (\forall i \in \overline{l},\overline{m}) \;
      (\Lo{{B_i}}, \Lo{{C_i}}, \hat{D}) &\in \dstxt
      \tag{by inversion on $D{\oplus}$} \\
      (\forall i \in \overline{l}) \;
      (\Lo{{A_i}}, \Lo{{C_i}}, \hat{D}) &\in \dstxt_{\Downarrow} \\
      \tag{by definition of $\dstxt_{\Downarrow}$ with
      $\Lo{{A_i}} \leq \Lo{{B_i}}$} \\
      (\forall i \in \overline{l}) \;
      (\Lo{{A_i}}, \Lo{{C_i}}, \hat{D}) &\in \dstxt'
      \tag{since $\dstxt_{\Downarrow} \subseteq \dstxt'$} \\
      (\intch{\caselist{l}{A}}, \intch{\caselist{l}{C}, \caselist{m}{C},
    \caselist{n}{C}}, \hat{D}) &\in F(\dstxt')
    \tag{by $D{\oplus}$}
  \end{flalign*}
\end{case}
$D{\&}$ follows a similar pattern.
\begin{case}
  $A = \downsl{A}$; similar to the proof of Lemma~\ref{lem:dsync-bigger}, there are three possible assignments for $B$ and $C$. We will present one of those subcases: let $B = \downsl{B}$ and $C = \downsl{C}$ with $\So{A} \leq \So{B} \leq \So{C}$. The other two cases are similar.
  \begin{flalign*}
    (\downsl{A}, \downsl{C}, \hat{D}) &\in \dstxt_{\Downarrow}, (\downsl{B}, \downsl{C}, \hat{D}) \in \dstxt \tag{this case} \\
    (\So{B}, \So{C}, \top) &\in \dstxt, \So{B} \leq \hat{D} \tag{by inversion on $D\Ldownsl$} \\
    (\So{A}, \So{C}, \top) &\in \dstxt_{\Downarrow} \tag{by definition of $\dstxt_{\Downarrow}$ with $\So{A} \leq \So{B}$} \\
    (\So{A}, \So{C}, \top) &\in \dstxt' \tag{since $\dstxt_{\Downarrow} \subseteq \dstxt'$} \\
    \So{A} \leq \hat{D}  \tag{because $\So{A} \leq \So{B} \leq \hat{D}$} \\
    (\downsl{A}, \downsl{C}, \hat{D}) &\in F{\dstxt'} \tag{by $D\Ldownsl$}
  \end{flalign*}
\end{case}
\begin{case}
  $A = \upls{A}$; again, there are three possible assignments for $B$ and $C$, and we will take the subcase when $B = \upls{B}$ and $C = \upls{C}$ with $\Lo{A} \leq \Lo{B} \leq \Lo{C}$. The other two cases are similar. This case finally uses our assumption that $(A, C, \hat{E}) \in \dstxt$ -- $\hat{E}$ must be $\top$ due to $A = \upls{A}$.
  \begin{flalign*}
    (\upls{A}, \upls{C}, \top) &\in \dstxt_{\Downarrow}, (\upls{B}, \upls{C}, \top) \in \dstxt \tag{this case} \\
    (\upls{A}, \upls{C}, \top) &\in \dstxt \tag{By assumption with $\hat{E} = \top$} \\
    (\Lo{A}, \Lo{C}, \upls{A}) &\in \dstxt \tag{by inversion on $D\Lupls$} \\
    (\Lo{A}, \Lo{C}, \upls{A}) &\in \dstxt' \tag{since $\dstxt \subseteq \dstxt'$} \\
    (\upls{A}, \upls{C}, \top) &\in F(\dstxt') \tag{by $D\Lupls$}
  \end{flalign*}
\end{case}
We missed one case, when $A = B = C = 1$, but this case is trivial since $\dstxt_{\Downarrow}$ does not add any new members to the set.
\end{proof}

\begin{lem}
  If $\dsyncl{A}{B}{\hat{C}}$ and $\dsyncl{A}{B}{\hat{D}}$, then \\${\dsyncl{A}{B}{\hat{C} \land \hat{D}}}$ for some $\Lo{A}, \Lo{B}, \hat{C},$ and $\hat{D}$.
\end{lem}
\begin{proof}
  First, recall that $\dsyncl{A}{B}{-}$ requires that $\Lo{A} \leq \Lo{B}$.  We want to show that
  $$\dstxt' \defeq \dstxt \cup \dstxt_{\land}$$
  is $F$-consistent with
  \begin{align*}
    \dstxt_{\land} \defeq \{(\Lo{A}, \Lo{B}, \hat{C} \land \hat{D})
      \defor (\Lo{A}, \Lo{B}, \hat{C}) \in \dstxt \land (\Lo{A}, \Lo{B}, \hat{D}) \in \dstxt
  \end{align*}
  As per usual, we will prove $F$-consistency of $\dstxt'$, that is, $\dstxt' \in F(\dstxt')$ by showing that each of the two sets $\dstxt$ and $\dstxt_{\land}$ are subsets of $F(\dstxt')$. \\
  $\dstxt \subseteq F(\dstxt')$ immediately follows from the same argument as in previous lemmas. \\
  We will now consider $\dstxt_{\land} \in F(\dstxt')$ by case analysis on the structure of $\Lo{A}$. We can infer the structure of $\Lo{B}$
  by inversion on the appropriate subtyping rule. For ease of presentation, let $\hat{E} = \hat{C} \land \hat{D}$; we will expand $\hat{E}$
  whenever necessary.
  \begin{case}
    $\Lo{A} = \upll{A'}$; then $\Lo{B} = \upll{B'}$ with $\Lo{A'} \leq \Lo{B'}$.
    \begin{flalign*}
      (\upll{A'}, \upll{B'}, \hat{E}) &\in \dstxt_{\land}, (\upll{A'}, \upll{B'}, \hat{C}) \in \dstxt, (\upll{A'}, \upll{B'}, \hat{D}) \in \dstxt \tag{this case} \\
      (\Lo{A'}, \Lo{B'}, \hat{C}) &\in \dstxt, (\Lo{A'}, \Lo{B'}, \hat{D}) \in \dstxt \tag{by inversion on $D\Lupll$} \\
      (\Lo{A'}, \Lo{B'}, \hat{E}) &\in \dstxt_{\land} \tag{by definition of $\dstxt_{\land}$} \\
      (\Lo{A'}, \Lo{B'}, \hat{E}) &\in \dstxt' \tag{since $\dstxt_{\land} \subseteq \dstxt'$} \\
      (\upll{A'}, \upll{B'}, \hat{E}) &\in F(\dstxt') \tag{by $D\Lupll$}
    \end{flalign*}
  \end{case}
  $\Ldownll, \otimes,$ and $\multimap$ follow a similar pattern of appealing to the continuation types. 
  \begin{case}
    $\Lo{A} = \intch{\caselist{l}{A}}$; then
    $\Lo{B} = \intch{\caselist{l}{B}, \caselist{m}{B}}$; with
    $\Lo{{A_i}} \leq \Lo{{B_i}} \; \forall i \in \overline{l}$.
    \begin{flalign*}
      (\intch{\caselist{l}{A}}, \intch{\caselist{l}{B}, \caselist{m}{B}}, \hat{E}) &\in \dstxt_{\land} \\
      (\intch{\caselist{l}{A}}, \intch{\caselist{l}{B}, \caselist{m}{B}}, \hat{C})
                                                                                   &\in \dstxt, (\intch{\caselist{l}{A}}, \intch{\caselist{l}{B}, \caselist{m}{B}},  \hat{D}) \in \dstxt \tag{this case} \\
      (\forall i \in \overline{l}) \; (\Lo{{A_i}}, \Lo{{B_i}}, \hat{C}) &\in \dstxt, (\Lo{{A_i}}, \Lo{{B_i}}, \hat{D}) \in \dstxt, \tag{by inversion on $D{\oplus}$} \\
      (\forall i \in \overline{l}) \; (\Lo{{A_i}}, \Lo{{B_i}}, \hat{E}) &\in \dstxt_{\land} \tag{by definition of $\dstxt_{\land}$} \\
      (\forall i \in \overline{l}) \; (\Lo{{A_i}}, \Lo{{B_i}}, \hat{E}) &\in \dstxt' \tag{since $\dstxt_{\land} \subseteq \dstxt'$} \\
      (\intch{\caselist{l}{A}}, \intch{\caselist{l}{B}, \caselist{m}{B}}, \hat{E}) &\in F(\dstxt') \tag{by $D{\oplus}$}
    \end{flalign*}
  \end{case}
  $D{\&}$ follows a similar pattern.
  \begin{case}
    $\Lo{A} = \downsl{A}$; then there are two subcases for the structure of $\Lo{B}$. We shall take the case when $\Lo{B} = \downsl{B}$ with $\So{A} \leq \So{B}$, but the other case, when $\Lo{B} = \downll{B'}$ follows a similar pattern. \\
    At this point we realize what $\hat{E}$ has to be -- either $\hat{E} = \bot$, in which case we want to derive a contradiction for this case (the $\bot$ constraint requires that there be no releases) or $\hat{E} = \So{E}$ meaning $\hat{E}$ is a non-trivial meet.
    \begin{subcase}
      $\hat{E} = \bot$.
      \begin{flalign*}
        (\downsl{A}, \downsl{B}, \bot) &\in \dstxt_{\land}, (\downsl{A}, \downsl{B}, \hat{C}) \in \dstxt, (\downsl{A}, \downsl{B}, \hat{D}) \in \dstxt \tag{this case} \\
        (\So{A}, \So{B}, \top) &\in \dstxt, \So{A} \leq \hat{C} \tag{by inversion on $D\Ldownsl$} \\
        (\So{A}, \So{B}, \top) &\in \dstxt, \So{A} \leq \hat{D} \tag{by inversion on $D\Ldownsl$} \\
        \text{Contradiction} \tag{since $\So{A}$ is a lower bound of $\hat{C} \land \hat{D}$ but $\So{A}$ is strictly greater than $\bot$}
      \end{flalign*}
    \end{subcase}
    \begin{subcase}
      $\hat{E} = \So{E}$ for some $\So{E}$.
      \begin{flalign*}
        (\downsl{A}, \downsl{B}, \So{E}) &\in \dstxt_{\land}, (\downsl{A}, \downsl{B}, \hat{C}) \in \dstxt, (\downsl{A}, \downsl{B}, \hat{D}) \in \dstxt \tag{this case} \\
        (\So{A}, \So{B}, \top) &\in \dstxt, \So{A} \leq \hat{C}, \So{A} \leq \hat{D} \tag{by inversion on $D\Ldownsl$} \\
        (\So{A}, \So{B}, \top) &\in \dstxt' \tag{since $\dstxt \subseteq \dstxt'$} \\
        (\downsl{A}, \downsl{C}, \So{E}) &\in F(\dstxt') \tag{by $D\Ldownsl$ with $\So{A} \leq \So{E}$ because $\So{A}$ is a lower bound of $\hat{C}$ and $\hat{D}$ and $\So{E}$ is the greatest lower bound}
      \end{flalign*}
    \end{subcase}
  \end{case}
  Unlike the previous lemmas, we require $\Lo{A}$ to be linear, so we do not need to consider $\Lupls$. The case when $A = B = 1$ is trivial.
\end{proof}

\section{Preservation Theorem}
\label{app:preservation-proof}
\begin{thm}[Preservation]
If $\JCb{\Gamma}{\Lambda}{\Theta}{\Delta}$ for some $\Lambda, \Theta, \Gamma,$ and $\Delta$, and $\Lambda;\Theta \rightarrow \Lambda';\Theta'$ for some $\Lambda';\Theta'$, then $\JCb{\Gamma'}{\Lambda'}{\Theta'}{\Delta}$ where $\Gamma' \preceq \Gamma$.
\end{thm}
\begin{proof}
By induction on the dynamics to construct a well-formed and well-typed configuration starting with $\JCb{\Gamma}{\Lambda}{\Theta}{\Delta}$.
\paragraph{Notation}
Many of the proof cases involve transitions between linear process terms (either proc or connect).
When reasoning with these transitions, we adopt the notation that $\Psi_a \to \Psi_a'$ that is, $\Psi_a$ represents the process term offering
$a$ before the transition and $\Psi_a'$ represents the process term offering $a$ after the transition.

\begin{case}
  \ref{dyn:fwdls}
  $$\procl{a}{\fwdls{a}{b}} \to \connect{a}{b}$$
  where $\Psi_a = \procl{a}{\fwdls{a}{b}}$ and $\Psi_a' = \connect{a}{b}$ (for the remaining cases, these metavariable assignments are implicit).
  Let $\Theta = \Theta_1, \Psi_a, \Theta_2$. Then by well-formedness, $\Lambda = \unavail{a}, \Lambda_1$.
  \begin{align*}
    &\JCb{\Gamma}{\Lambda}{\Theta_1, \Psi_a, \Theta_2}{\Delta}	\tag{assumption} \\
    &\JCs{\Gamma}{\Lambda}{\Gamma} \quad \JCl{\Gamma}{\Theta_1, \Psi_a, \Theta_2}{\Delta}	\tag{by inversion on $\Omega$} \\
    &\JCl{\Gamma}{\procl{a}{\fwdls{a}{b}}, \Theta_2}{\Lt{a}{A}, \Delta_p} \tag{by Lemma~\ref{lem:cfg-peel} and expanding $\Psi_a$} \\
    &\JCl{\Gamma}{\Theta_2}{\Delta_p} \quad \JTl{\Gamma}{\cdot}{\fwdls{a}{b}}{a}{A'} \tag{by inversion on $\Theta 3$} \\
    &\Sto{b}{\hat{B}} \in \Gamma \quad \hat{B} \leq \Lo{A'} \tag{by inversion on $ID_{\scriptscriptstyle {LS}}$} \\
    &\hat{B} \leq \Lo{A} \tag{by transitivity of $\leq$} \\
    &\JCpc{\Gamma}{a}{b}{\Lo{A}}{\Theta_2}{\Delta_p} \tag{by $\Theta2$} \\
    &\JCl{\Gamma}{\Theta_1, \Psi_a', \Theta_2}{\Delta} \tag{by Lemma~\ref{lem:cfg-subs}} \\
    &\JCb{\Gamma}{\Lambda}{\Theta_1, \Psi_a', \Theta_2}{\Delta}	\tag{by $\Omega$}
  \end{align*}
  The well-formedness conditions are maintained because only $\Psi_a \in \Theta$ was replaced by $\Psi_a'$.
\end{case}

\begin{case}
  \ref{dyn:extch}
  $$\procl{a}{b.i; P}, \procl{b}{\casep{b}{\caselistp{l}{Q}, \caselistp{m}{Q}}} \to \procl{a}{P}, \procl{b}{Q_i} \quad (i \in \overline{l})$$
  Then $\Theta = \Theta_1, \Psi_a, \Theta_2, \Psi_b, \Theta_3$
  \begin{align*}
    &\JCb{\Gamma}{\Lambda}{\Theta_1, \Psi_a, \Theta_2, \Psi_b, \Theta_3}{\Delta} \tag{assumption} \\
    &\JCs{\Gamma}{\Lambda}{\Gamma} \quad \JCl{\Gamma}{\Theta_1, \Psi_a, \Theta_2, \Psi_b, \Theta_3}{\Delta} \tag{by inversion on $\Omega$} \\
    &\JCl{\Gamma}{\Psi_a, \Theta_2, \Psi_b, \Theta_3}{\Lt{a}{A}, \Delta_r} \tag{by Lemma~\ref{lem:cfg-peel}} \\
    &\JCl{\Gamma}{\Psi_a, \Psi_b, \Theta_r}{\Lt{a}{A}, \Delta_r} \tag{by Lemma~\ref{lem:cfg-permutation} and $\Theta_r = \Theta_2, \Theta_3$} \\
    &\JCl{\Gamma}{\Psi_b, \Theta_r}{\Lto{b}{\extch{\caselist{l}{B}}}, \Delta_a, \Delta_r} \quad \JTl{\Gamma}{\Delta_a}{b.i; P}{a}{A'} \\
    &\quad\Sto{a}{\hat{A}} \in \Gamma \quad \dsyncl{A'}{A}{\hat{A}} \tag{by inversion on $\Theta 3$} \\
    &\JCl{\Gamma}{\Theta_r}{\Delta_a, \Delta_b, \Delta_r} \quad
      \JTlr{\Gamma}{\Delta_b}{\casep{b}{\caselistp{l}{Q}, \caselistp{m}{Q}}}{b}{\extch{\caselist{l}{B'}, \caselist{m}{B'}}} \\
    &\quad\Sto{b}{\hat{B}} \in \Gamma \quad \dsync{\extch{\caselist{l}{B'}, \caselist{m}{B'}}}{\extch{\caselist{l}{B}}}{\hat{B}} \tag{by inversion on $\Theta 3$} \\
    &\JTl{\Gamma}{\Delta_b}{Q_i}{b}{{B_i}'} \tag{inversion on ${\&}R$} \\
    &\Lo{{B_i}'} \leq \Lo{{B_i}} \quad \dsyncl{{B_i'}}{{B_i}}{\hat{B}} \tag{by inversion on $\leq_{\&}$ and E{\&} respectively} \\
    &\JCl{\Gamma}{\Psi_b', \Theta_r}{\Lt{b}{{B_i}}, \Delta_r, \Delta_a} \tag{by $\Theta 3$} \\
    &\JTl{\Gamma}{\Delta_a, \Lt{b}{{B_i}}}{P}{a}{A'} \tag{inversion on ${\&}L$} \\
    &\JCl{\Gamma}{\Psi_a', \Psi_b', \Theta_r}{\Lt{a}{A}, \Delta_r} \tag{by $\Theta 3$} \\
    &\JCl{\Gamma}{\Theta_1, \Psi_a', \Psi_b', \Theta_r}{\Delta} \tag{by Lemma~\ref{lem:cfg-subs}} \\
    &\JCb{\Gamma}{\Lambda}{\Theta_1, \Psi_a', \Psi_b', \Theta_r}{\Delta} \tag{by $\Omega$}
  \end{align*}
  The well-formedness conditions are maintained because $\Psi_a$ and $\Psi_b$ were replaced by $\Psi_a'$ and $\Psi_b'$ respectively in $\Theta$.
\end{case}

The proof of \ref{dyn:intch} is similar to \ref{dyn:extch}.

\begin{case}
  \ref{dyn:tensor}
  $$\procl{a}{\asgn{\Lo{y}}{\recv{b}}; P}, \procl{b}{\sendl{b}{c}; Q}, \Psi_c \to \procl{a}{[\Lsub{c}{y}]P}, \procl{b}{Q}, \Psi_c$$
  Then $\Theta = \Theta_1, \Psi_a, \Theta_2, \Psi_b, \Theta_3, \Psi_c, \Theta_4$.
  \begin{align*}
    &\JCb{\Gamma}{\Lambda}{\Theta_1, \Psi_a, \Theta_2, \Psi_b, \Theta_3, \Psi_c, \Theta_4}{\Delta} \tag{assumption} \\
    &\JCs{\Gamma}{\Lambda}{\Gamma} \quad \JCl{\Gamma}{\Theta_1, \Psi_a, \Theta_2, \Psi_b, \Theta_3, \Psi_c, \Theta_4}{\Delta} \tag{by inversion on $\Omega$} \\
    &\JCl{\Gamma}{\Psi_a, \Theta_2, \Psi_b, \Theta_3, \Psi_c, \Theta_4}{\Lt{a}{A}, \Delta_r} \tag{by Lemma~\ref{lem:cfg-peel}} \\
    &\JCl{\Gamma}{\Psi_a, \Psi_b, \Psi_c, \Theta_r}{\Lt{a}{A}, \Delta_r} \tag{by Lemma~\ref{lem:cfg-permutation} and $\Theta_r = \Theta_2, \Theta_3, \Theta_4$} \\
    &\JCl{\Gamma}{\Psi_b, \Psi_c, \Theta_r}{\Lto{b}{\tensor{C^a}{B}}, \Delta_a, \Delta_r} \quad
      \JTl{\Gamma}{\Delta_a, \Lt{b}{B}}{\asgn{\Lo{y}}{\recv{b}}; P}{a}{A'} \\
    &\Sto{a}{\hat{A}} \in \Gamma \quad \dsyncl{A'}{A}{\hat{A}} \tag{by inversion on $\Theta 3$} \\
    &\JTl{\Gamma}{\Delta_a, \Lt{b}{B}, \Lt{c}{C^a}}{[\Lsub{c}{y}]P}{a}{A'} \tag{by inversion on ${\otimes}L$ and $\alpha$ equivalance} \\
    &\JCl{\Gamma}{\Psi_c, \Theta_r}{\Lt{c}{C}, \Delta_a, \Delta_b, \Delta_r} \quad
      \JTlr{\Gamma}{\Delta_b, \Lt{c}{C}}{\sendl{b}{c}; Q}{b}{\tensor{C^b}{B'}} \\
    &\Sto{b}{\hat{B}} \in \Gamma \quad \dsync{\tensor{C^b}{B'}}{\tensor{C^a}{B}}{\hat{B}} \tag{by inversion on $\Theta 3$} \\
    &\JTl{\Gamma}{\Delta_b}{Q}{b}{B'} \quad \Lo{C} \leq \Lo{C^b} \tag{by inversion on ${\otimes}R$} \\
    &\Lo{C^b} \leq \Lo{C^a} \quad \Lo{B'} \leq \Lo{B} \quad \dsyncl{B'}{B}{\hat{B}} \tag{by inversion on $\leq_{\otimes}$ and $E{\otimes}$ respectively} \\
    &\JCl{\Gamma}{\Psi_c, \Theta_r}{\Lt{c}{C^a}, \Delta_a, \Delta_b, \Delta_r} \tag{by Lemma~\ref{lem:cfg-bigger} since $\Lo{C} \leq \Lo{C^b} \leq \Lo{C^a}$.} \\
    &\JCl{\Gamma}{\Psi_b', \Psi_c, \Theta_r}{\Lt{b}{B}, \Lt{c}{C^a}, \Delta_a, \Delta_r} \tag{by $\Theta3$} \\
    &\JCl{\Gamma}{\Psi_a', \Psi_b', \Psi_c, \Theta_r}{\Lt{a}{A}, \Delta_r} \tag{by $\Theta3$} \\
    &\JCl{\Gamma}{\Theta_1, \Psi_a', \Psi_b', \Psi_c, \Theta_r}{\Delta} \tag{by Lemma~\ref{lem:cfg-subs}} \\
    &\JCb{\Gamma}{\Lambda}{\Theta_1, \Psi_a', \Psi_b', \Psi_c, \Theta_r}{\Delta} \tag{by $\Omega$}
  \end{align*}
  The well-formedness conditions are maintained because $\Psi_a$ and $\Psi_b$ were replaced by $\Psi_a'$ and $\Psi_b'$ respectively in $\Theta$.
\end{case}

\begin{case}
  \ref{dyn:tensor2}
  \begin{align*}
    &\procl{a}{\asgn{\Lo{y}}{\recv{b}}; P}, \procl{b}{\sends{b}{c}; Q} \\
    \to \quad &\procl{a}{[\Lsub{d}{y}]P}, \procl{b}{Q}, \connect{d}{c}, \unavail{d} \fresh{d}
  \end{align*}
  Then $\Theta = \Theta_1, \Psi_a, \Theta_2, \Psi_b, \Theta_3$.
  \begin{align*}
    &\JCb{\Gamma}{\Lambda}{\Theta_1, \Psi_a, \Theta_2, \Psi_b, \Theta_3}{\Delta} \tag{assumption} \\
    &\JCs{\Gamma}{\Lambda}{\Gamma} \quad \JCl{\Gamma}{\Theta_1, \Psi_a, \Theta_2, \Psi_b, \Theta_3}{\Delta} \tag{by inversion on $\Omega$} \\
    &\JCl{\Gamma}{\Psi_a, \Theta_2, \Psi_b, \Theta_3}{\Lt{a}{A}, \Delta_r} \tag{by Lemma~\ref{lem:cfg-peel}} \\
    &\JCl{\Gamma}{\Psi_a, \Psi_b, \Theta_r}{\Lt{a}{A}, \Delta_r} \tag{by Lemma~\ref{lem:cfg-permutation} and $\Theta_r = \Theta_2, \Theta_3$} \\
    &\JCl{\Gamma}{\Psi_b, \Theta_r}{\Lto{b}{\tensor{C^a}{B}}, \Delta_a, \Delta_r} \quad
      \JTl{\Gamma}{\Delta_a, \Lt{b}{B}}{\asgn{\Lo{y}}{\recv{b}}; P}{a}{A'} \\
    &\Sto{a}{\hat{A}} \in \Gamma \quad \dsyncl{A'}{A}{\hat{A}} \tag{by inversion on $\Theta 3$} \\
    &\JTl{\Gamma}{\Delta_a, \Lt{b}{B}, \Lt{d}{C^a}}{[\Lsub{d}{y}]P}{a}{A'} \tag{by inversion on ${\otimes}L$ and $\alpha$ equivalance} \\
    &\JCl{\Gamma}{\Theta_r}{\Delta_a, \Delta_b, \Delta_r} \quad
      \JTlr{\Gamma}{\Delta_b, \Lt{c}{C}}{\sends{b}{c}; Q}{b}{\tensor{C^b}{B'}} \\
    &\Sto{b}{\hat{B}} \in \Gamma \quad \dsync{\tensor{C^b}{B'}}{\tensor{C^a}{B}}{\hat{B}} \tag{by inversion on $\Theta 3$} \\
    &\JTl{\Gamma}{\Delta_b}{Q}{b}{B'} \quad \hat{C} \leq \Lo{C^b} \tag{by inversion on ${\otimes}\So{R}$} \\
    &\JCpc{\Gamma}{d}{c}{\Lo{C^a}}{\Theta_r}{\Delta_a, \Delta_b, \Delta_r} \\
    &\JCl{\Gamma}{\Psi_b', \Psi_d, \Theta_r}{\Lt{b}{B}, \Lt{d}{C^a}, \Delta_a, \Delta_r} \tag{by $\Theta3$ where $\Psi_d = \connect{d}{c}$} \\
    &\JCl{\Gamma}{\Psi_a', \Psi_b', \Psi_d, \Theta_r}{\Lt{a}{A}, \Delta_r} \tag{by $\Theta3$} \\
    &\JCl{\Gamma}{\Theta_1, \Psi_a', \Psi_b', \Psi_d, \Theta_r}{\Delta} \tag{by Lemma~\ref{lem:cfg-subs}} \\
    &\JCl{\Gamma'}{\Theta_1, \Psi_a', \Psi_b', \Psi_d, \Theta_r}{\Delta} \tag{by Lemma~\ref{lem:cfg-stable} with $\Gamma' = \Gamma, \Sto{d}{\bot}$} \\
    &\JCs{\Gamma'}{\Lambda}{\Gamma} \tag{by Lemma~\ref{lem:cfg-stable}} \\
    &\JCs{\Gamma'}{\unavail{d}}{\Sto{d}{\bot}} \tag{by $\Lambda4$} \\
    &\JCs{\Gamma'}{\Lambda, \unavail{d}}{\Gamma'} \tag{by $\Lambda2$} \\
    &\JCb{\Gamma'}{\Lambda, \unavail{d}}{\Theta_1, \Psi_a', \Psi_b', \Psi_d, \Theta_r}{\Delta} \tag{by $\Omega$}
  \end{align*}
  The well-formedness conditions are maintained because $\Psi_a$ and $\Psi_b$ were replaced by $\Psi_a'$ and $\Psi_b'$ respectively in $\Theta$ and a $\Psi_d$ was added in $\Theta$ where $d$ is fresh along with a corresponding $\unavail{d}$ in $\Lambda' = \Lambda, \unavail{d}$.
\end{case}

The proofs of \ref{dyn:loli} and \ref{dyn:loli2} are similar to \ref{dyn:tensor} and \ref{dyn:tensor2} respectively.\\
We will now present some of the harder cases:
\begin{case}
  \ref{dyn:fwdll}
  $$\procl{a}{\fwdll{a}{b}}, \Psi_b \to \Psi_b \quad (\Lo{a} := \Lo{b}, \So{a} := \So{b})$$
  Then $\Theta = \Theta_1, \Psi_a, \Theta_2, \Psi_b, \Theta_3$ and $\Lambda = \unavail{a}, \unavail{b}, \Lambda_1$ by Lemma~\ref{lem:cfg-balance}.
  \begin{align*}
    &\JCb{\Gamma}{\Lambda}{\Theta_1, \Psi_a, \Theta_2, \Psi_b, \Theta_3}{\Delta} \tag{assumption} \\
    &\JCs{\Gamma}{\Lambda}{\Gamma} \quad \JCl{\Gamma}{\Theta_1, \Psi_a, \Theta_2, \Psi_b, \Theta_3}{\Delta} \tag{by inversion on $\Omega$} \\
    &\JCl{\Gamma}{\Psi_a, \Theta_2, \Psi_b, \Theta_3}{\Lt{a}{A}, \Delta_r} \tag{by Lemma~\ref{lem:cfg-peel}} \\
    &\JCl{\Gamma}{\Psi_a, \Psi_b, \Theta_r}{\Lt{a}{A}, \Delta_r} \tag{by Lemma~\ref{lem:cfg-permutation} and $\Theta_r = \Theta_2, \Theta_3$} \\
    &\JCl{\Gamma}{\Psi_b, \Theta_r}{\Lto{b}{\Lt{b}{B}}, \Delta_r} \; \JTl{\Gamma}{\Lt{b}{B}}{\fwdll{a}{b}}{a}{A'} \;
      \Sto{a}{\hat{A}} \in \Gamma \; \dsyncl{A'}{A}{\hat{A}} \tag{by inversion on $\Theta 3$} \\
    &\Lo{B} \leq \Lo{A'} \tag{by inversion on $\Lo{ID}$}
  \end{align*}
  At this point we need to case on the structure of $\Psi_b$. In both cases we will show that $\JCl{\Gamma'}{\Psi_a', \Theta_r}{\Lt{a}{A}, \Delta_r}$ for some $\Gamma' \preceq \Gamma$ and $\Psi_a'$ being directly defined from $\Psi_b$.
  \begin{subcase}
    $\Psi_b = \connect{b}{c}$ for some $\So{c}$.
    \begin{align*}
      &\JCpc{\Gamma}{b}{c}{\Lo{B}}{\Theta_r}{\Delta_r} \quad \Sto{c}{\hat{C}} \in \Gamma \quad \hat{C} \leq \Lo{B} \tag{by inversion on $\Theta2$} \\
      &\hat{C} \leq \Lo{A} \tag{by transitivity of $\leq$} \\
      &\JCpc{\Gamma}{b}{c}{\Lo{A}}{\Theta_r}{\Delta_r} \tag{by $\Theta2$} \\
      &\JCpc{\Gamma}{a}{c}{\Lo{A}}{\Theta_r}{\Delta_r} \tag{from renaming}
    \end{align*}
  \end{subcase}
  \begin{subcase}
    $\Psi_b = \procl{b}{P}$ for some process term $P$.
    \begin{align*}
      &\JCl{\Gamma}{\Theta_r}{\Delta_b, \Delta_r} \quad \JTl{\Gamma}{\Delta_b}{P}{b}{B'} \quad \Sto{b}{\hat{B}} \in \Gamma \quad \dsyncl{B'}{B}{\hat{B}} \tag{by inversion on $\Theta 3$} \\
      &\Lo{B'} \leq \Lo{B} \leq \Lo{A'} \leq \Lo{A} \\
      &\dsyncl{B'}{A}{\hat{B}} \quad \dsyncl{B'}{A}{\hat{A}} \tag{by Lemma~\ref{lem:dsync-bigger} and Lemma~\ref{lem:dsync-smaller} respectively} \\
      &\dsyncl{B'}{A}{\hat{B} \land \hat{A}} \tag{by Lemma~\ref{lem:dsync-meet}} \\ 
      &\JCl{\Gamma'}{\Theta_r}{\Delta_b, \Delta_r} \tag{by Lemma~\ref{lem:cfg-stable} with $\Gamma' = [\Sto{a}{\hat{B} \land \hat{A}}/\Sto{a}{\hat{A}}]\Gamma$} \\
      &\JTl{\Gamma'}{\Delta_b}{P}{b}{B'} \tag{by Lemma~\ref{lem:static-stable}} \\
      &\JTl{\Gamma'}{\Delta_b}{[\Lsub{a}{b}, \Ssub{a}{b}]P}{a}{B'} \tag{by $\alpha$ equivalence for $\Lsub{a}{b}$ and a combination of $\alpha$ equivalence and Lemma~\ref{lem:cfg-stable} for $\Ssub{a}{b}$} \\
      &\JCpp{\Gamma'}{a}{[\Lsub{a}{b}, \Ssub{a}{b}]P}{\Lo{A}}{\Theta_r}{\Delta_r} \tag{by $\Theta3$}
    \end{align*}
  \end{subcase}
  We will now continue assuming $\JCl{\Gamma'}{\Psi_a', \Theta_r}{\Lt{a}{A}, \Delta_r}$ with $\Gamma' \preceq \Gamma$ and $\Psi_a' = [\Lsub{a}{b}, \Ssub{a}{b}]\Psi_b$. For the connect case that did not require a smaller $\Gamma$, simply set $\Gamma' = \Gamma$ since $\Gamma' \preceq \Gamma$ by reflexivity.
  \begin{align*}
    &\JCl{\Gamma'}{\Theta_1, \Psi_a, \Theta_r}{\Delta} \tag{by Lemma~\ref{lem:cfg-stable}} \\
    &\JCl{\Gamma'}{\Theta_1, \Psi_a', \Theta_r}{\Delta} \tag{by Lemma~\ref{lem:cfg-subs}} \\
    &\JCs{\Gamma'}{\Lambda}{\Gamma} \tag{by Lemma~\ref{lem:cfg-subs}} \\
    &\JCs{\Gamma'}{\unavail{a}}{\Sto{a}{\bot}} \tag{by $\Lambda4$} \\
    &\JCs{\Gamma'}{\unavail{b}, \Theta_1}{\Gamma''} \tag{by inversion on $\Lambda2$ where $\Gamma' = \Gamma'', \Sto{a}{\bot}$} \\
    &\JCs{\Gamma'}{\Lambda}{\Gamma'} \tag{by $\Lambda2$} \\
    &\JCs{\Gamma'}{[\Ssub{a}{b}]\Lambda}{\Gamma'} \tag{by $\alpha$ equivalence} \\ 
    &\JCl{\Gamma'}{[\Ssub{a}{b}]\Theta_1, \Psi_a', [\Ssub{a}{b}]\Theta_r}{\Delta} \tag{by $\alpha$ equivalence} \\
    &\JCb{\Gamma'}{\Lambda}{[\Ssub{a}{b}, \Lsub{a}{b}]\Theta_1, \Psi_a', [\Ssub{a}{b}]\Theta_r}{\Delta} \tag{by $\Omega$}
  \end{align*}
  Well-formedness is easily maintained because we only removed something from the linear fragment (it is okay to have dangling unavail terms in the shared fragment).
\end{case}

\begin{case}
  \ref{dyn:upls}
  \begin{align*}
    &\procl{a}{\asgn{\Lo{x}}{\acqs{b}}; P}, \procs{b}{\asgn{\Lo{x}}{\accs{b}}; Q} \\
    \to \quad &\procl{a}{[\Lsub{b}{x}]P}, \procl{b}{[\Lsub{b}{x}]Q}, \unavail{b}
  \end{align*}
  Then $\Lambda = \Lambda_b, \Lambda_1$ and $\Theta = \Theta_1, \Psi_a, \Theta_2$ with $\Lambda_b = \procs{b}{\asgn{\Lo{x}}{\accs{b}}; Q}$.\\
  We also define $\Psi_b' = \procl{b}{[\Lsub{b}{x}]Q}$.
  \begin{align*}
    &\JCb{\Gamma}{\Lambda_b, \Lambda_1}{\Theta_1, \Psi_a, \Theta_2}{\Delta}	\tag{assumption} \\
    &\JCs{\Gamma}{\Lambda_b, \Lambda_1}{\Gamma} \quad \JCl{\Gamma}{\Theta_1, \Psi_a, \Theta_2}{\Delta}	\tag{by inversion on $\Omega$} \\
    &\JCs{\Gamma}{\Lambda_b}{\Sto{b}{\upls{B}}} \quad \JCs{\Gamma}{\Lambda_1}{\Gamma'} \tag{by inversion on $\Lambda2$ with $\Gamma = \Sto{b}{\upls{B}}, \Gamma'$} \\
    &\dsync{\upls{B'}}{\upls{B}}{\top} \quad \JTsr{\Gamma}{\asgn{\Lo{x}}{\accs{b}}; Q}{b}{\upls{B'}} \tag{by inversion on $\Lambda3$} \\
    &\JTl{\Gamma}{\cdot}{[\Lsub{b}{x}]Q}{b}{B'} \tag{by inversion on $\Lupls R$ and $\alpha$ equivalence} \\
    &\dsyncl{B'}{B}{\upls{B'}} \tag{by inversion on $D\Lupls$} \\
    &\JCl{\Gamma}{\Psi_a, \Theta_2}{\Lt{a}{A}, \Delta_p} \tag{by Lemma~\ref{lem:cfg-peel}} \\
    &\JCl{\Gamma}{\Theta_2}{\Delta_a, \Delta_p} \quad \JTl{\Gamma}{\Delta_a}{\asgn{\Lo{x}}{\acqs{b}}}{a}{A'} \\
    &\Sto{a}{\hat{A}} \in \Gamma \quad \dsyncl{A'}{A}{\hat{A}} \tag{by inversion on $\Theta 3$} \\
    &\JTl{\Gamma}{\Delta_a, \Lt{b}{B^a}}{[\Lsub{b}{x}]P}{a}{A'} \quad \upls{B} \leq \upls{B^a} \tag{by inversion on $\Lupls L$ and $\alpha$ equivalence} \\
    &\JCl{\Gamma}{\Psi_b', \Theta_2}{\Lt{b}{B}, \Delta_a, \Delta_p} \tag{by $\Lambda3$} \\
    &\JCl{\Gamma}{\Psi_b', \Theta_2}{\Lt{b}{B^a}, \Delta_a, \Delta_p} \tag{by Lemma~\ref{lem:cfg-bigger}} \\
    &\JCl{\Gamma}{\Psi_a', \Psi_b', \Theta_2}{\Lto{a}{A}, \Delta_p} \tag{by $\Theta3$} \\
    &\JCl{\Gamma}{\Theta_1, \Psi_a', \Psi_b', \Theta_2}{\Delta} \tag{by Lemma~\ref{lem:cfg-subs}} \\
    &\JCs{\Gamma}{\unavail{b}}{\Sto{b}{\upls{B}}} \tag{by $\Lambda4$} \\
    &\JCs{\Gamma}{\unavail{b}, \Lambda_1}{\Gamma} \tag{by $\Lambda2$} \\
    &\JCb{\Gamma}{\Lambda}{\Theta_1, \Psi_a', \Psi_b', \Theta_2}{\Delta}  \tag{by $\Omega$}
  \end{align*}
  Well-formedness is maintained because $\Psi_b \notin \Theta$ and there is a corresponding $\unavail{b}$ to the newly added $\Psi_b'$.
\end{case}

\begin{case}
  \ref{dyn:upls2}
  \begin{align*}
    &\procl{a}{\asgn{\Lo{x}}{\acql{b}}; P}, \connect{b}{c}, \procs{c}{\asgn{\Lo{x}}{\accs{c}}; Q} \\
    \to \quad &\procl{a}{[\Lsub{c}{x}]P}, \procl{c}{[\Lsub{c}{x}]Q},	\unavail{c}
  \end{align*}
  Then $\Lambda = \Lambda_c, \Lambda_1$ and $\Theta = \Theta_1, \Psi_a, \Theta_2, \Psi_b, \Theta_3$ with $\Lambda_c = \procs{c}{\asgn{\Lo{x}}{\accs{c}}; Q}$.\\
  We also define $\Psi_c' = \procl{c}{[\Lsub{c}{x}]Q}$.
  \begin{align*}
    &\JCb{\Gamma}{\Lambda_c, \Lambda_1}{\Theta_1, \Psi_a, \Theta_2, \Psi_b, \Theta_3}{\Delta}	\tag{assumption} \\
    &\JCs{\Gamma}{\Lambda_c, \Lambda_1}{\Gamma} \quad \JCl{\Gamma}{\Theta_1, \Psi_a, \Theta_2, \Psi_b, \Theta_3}{\Delta}	\tag{by inversion on $\Omega$} \\
    &\JCs{\Gamma}{\Lambda_c}{\Sto{c}{\upls{C}}} \quad \JCs{\Gamma}{\Lambda_1}{\Gamma'} \tag{by inversion on $\Lambda2$ with $\Gamma = \Sto{c}{\upls{C}}, \Gamma'$} \\
    &\dsync{\upls{C'}}{\upls{C}}{\top} \quad \JTsr{\Gamma}{\asgn{\Lo{x}}{\accs{c}}; Q}{c}{\upls{C'}} \tag{by inversion on $\Lambda3$} \\
    &\JTl{\Gamma}{\cdot}{[\Lsub{c}{x}]Q}{c}{C'} \tag{by inversion on $\Lupls R$ and $\alpha$ equivalence} \\
    &\dsyncl{C'}{C}{\upls{C'}} \tag{by inversion on $D\Lupls$} \\
    &\JCl{\Gamma}{\Psi_a, \Theta_2, \Psi_b, \Theta_3}{\Lt{a}{A}, \Delta_p} \tag{by Lemma~\ref{lem:cfg-peel}} \\
    &\JCl{\Gamma}{\Psi_a, \Psi_b, \Theta_r}{\Lt{a}{A}, \Delta_p} \tag{by Lemma~\ref{lem:cfg-permutation} with $\Theta_r = \Theta_2, \Theta_3$} \\
    &\JCl{\Gamma}{\connect{b}{c}, \Theta_2}{\Sto{b}{\upll{B}}, \Delta_a, \Delta_p} \quad \JTl{\Gamma}{\Delta_a, \Sto{b}{\upll{B}}}{\asgn{\Lo{x}}{\acql{b}}}{a}{A'} \\
    &\Sto{a}{\hat{A}} \in \Gamma \quad \dsyncl{A'}{A}{\hat{A}} \tag{by inversion on $\Theta 3$} \\
    &\JCl{\Gamma}{\Theta_r}{\Delta_a, \Delta_p} \quad \upls{C} \leq \upll{B} \tag{by inversion on $\Theta 2$} \\
    &\Lo{C} \leq \Lo{B} \quad \dsyncl{C'}{B}{\upls{C'}} \tag{by inversion on $\leq_{\Lupls\Lupll}$ and Lemma~\ref{lem:dsync-bigger} respectively} \\
    &\JCl{\Gamma}{\Psi_c', \Theta_r}{\Lt{c}{C}, \Delta_a, \Delta_p} \tag{by $\Lambda3$} \\
    &\JTl{\Gamma}{\Delta_a, \Lt{c}{C}}{[\Lsub{c}{x}]P}{a}{A'} \tag{by inversion on $\Lupls L$ and $\alpha$ equivalence} \\
    &\JCl{\Gamma}{\Psi_a', \Psi_c', \Theta_2}{\Lto{a}{A}, \Delta_p} \tag{by $\Theta3$} \\
    &\JCl{\Gamma}{\Theta_1, \Psi_a', \Psi_b', \Theta_2}{\Delta} \tag{by Lemma~\ref{lem:cfg-subs}} \\
    &\JCs{\Gamma}{\unavail{c}}{\Sto{c}{\upls{C}}} \tag{by $\Lambda4$} \\
    &\JCs{\Gamma}{\unavail{c}, \Lambda_1}{\Gamma} \tag{by $\Lambda2$} \\ 
    &\JCb{\Gamma}{\Lambda}{\Theta_1, \Psi_a', \Psi_c', \Theta_2}{\Delta}  \tag{by $\Omega$}
  \end{align*}
  Well-formedness is maintained because $\Psi_c \notin \Theta$ and there is a corresponding $\unavail{c}$ to the newly added $\Psi_c'$.
\end{case}
Other omitted cases follow a similar strategy as presented.
\end{proof}

\section{Progress Theorem}
\label{app:progress-proof}
\begin{thm}[Progress]
  If $\JCb{\Gamma}{\Lambda}{\Theta}{\Delta}$ then either:
  \begin{enumerate}
    \item $\Lambda;\Theta \rightarrow \Lambda';\Theta$ for some $\Lambda'$ or
    \item $\Lambda$ is poised and one of:
      \begin{enumerate}
        \item $\Lambda;\Theta \rightarrow \Lambda';\Theta'$ or
        \item $\Theta$ is poised or
        \item a linear process in $\Theta$ is stuck and therefore unable to acquire
      \end{enumerate}
  \end{enumerate}
\end{thm}
\begin{proof}
  \begin{align*}
    \JCb{\Gamma}{\Lambda}{\Theta}{\Delta} \tag{by assumption} \\
    \JCs{\Gamma}{\Lambda}{\Gamma} \quad \JCl{\Gamma}{\Theta}{\Delta} \tag{by inversion on $\Omega$}
  \end{align*}
  for some $\Gamma, \Lambda, \Theta,$ and $\Delta$.

  We first show that either $\Lambda \to \Lambda'$ for some $\Lambda'$ or that $\Lambda$ is poised by induction on the derivation of $\JCs{\Gamma}{\Lambda}{\Gamma}$.
  \begin{case}
    $$\infer[\Lambda 1]{\JCs{\Gamma}{\cdot}{\cdot}}{}$$
    $(\cdot)$ is poised since there is no proc term.
  \end{case}
  \begin{case}
    $$\infer[\Lambda 2]{\JCs{\Gamma}{\Lambda_1, \Lambda_2}{\Gamma_1, \Gamma_2}}{\JCs{\Gamma}{\Lambda_1}{\Gamma_1} & \JCs{\Gamma}{\Lambda_2}{\Gamma_2}}$$
    Then either $\Lambda_1 \to \Lambda_1'$ or $\Lambda_1$ is poised by
  IH, and similarly, either $\Lambda_2 \to \Lambda_2'$ or
$\Lambda_2$ is poised by IH. If both $\Lambda_1$ and
$\Lambda_2$ are poised, then the concatenation $\Lambda_1, \Lambda_2$ is
poised. Otherwise, we take the concatenation of the components that progresses.
In particular, if $\Lambda_1 \to \Lambda_1'$ and $\Lambda_2$ is poised,
$\Lambda_1, \Lambda_2 \to \Lambda_1', \Lambda_2$ (and similarly for the other
two combinations).
  \end{case}
  \begin{case}
    $$\infer[\Lambda 3]{\JCs{\Gamma}{\procs{a}{P}}{\St{a}{A}}}{\dsyncs{A'}{A}{\top} & \JTs{\Gamma}{P}{a}{A'}}$$
    We proceed by case analysis on the syntactic form of $P$ inferred from inversion on the appropriate typing rule on the derivation of $\JTs{\Gamma}{P}{a}{A'}$.
    \begin{subcase}
      $P = \fwdss{a}{b}$. This case requires a global substitution on the top level $\Lambda$. Since there is no ordering constraint on $\Lambda$, let $\Lambda = \procs{a}{\fwdss{a}{b}}, \Lambda_1$ without loss of generality. Then by \ref{dyn:fwdss}, 
      $$\Lambda \to [\Ssub{a}{b}]\Lambda_1$$
    \end{subcase}
    \begin{subcase}
      ${P = \spawn{\So{x}}{\So{X}}{\Solist{b}}; Q}$, then by \ref{dyn:spawnss}, 
      $$\procs{a}{\spawn{\So{x}}{\So{X}}{\Solist{b}}; Q} \to \procs{a}{[\Ssub{c}{x}]Q}, \procs{c}{[\Ssub{c}{x'}, \Ssublist{b}{y'}]P} \fresh{c}$$
    \end{subcase}
    \begin{subcase}
      ${P = \asgn{\Lo{a}}{\accs{a}}; Q}$, then $\procs{a}{P}$ is poised by definition.
    \end{subcase}
  \end{case}
  \begin{case}
    $$\infer[\Lambda 4]{\JCs{\Gamma}{\unavail{a}}{\Sto{a}{\hat{A}}}}{}$$
    $\unavail{a}$ is poised since there is no proc term.
  \end{case}
  \setcounter{case}{0}

  That concludes the first part of the proof. Now to show the second part, we will assume that $\Lambda$ is poised and proceed by induction on the derivation of $\JCl{\Gamma}{\Theta}{\Delta}$ to show one of:
  \begin{enumerate}[label=(\alph*)]
    \item $\Lambda;\Theta \to \Lambda';\Theta'$ for some $\Lambda'$ and $\Theta'$
    \item $\Theta$ poised
    \item some $\Psi \in \Theta$ is stuck
  \end{enumerate}
  \noindent We will showcase the style of the proof along with the interesting cases.
  \begin{case}
    $$\infer[\Theta 1]{\JCl{\Gamma}{\cdot}{\cdot}}{}$$
    $(\cdot)$ is poised since there is no proc term.
  \end{case}
  \begin{case}
    $$\infer[\Theta 2]{\JCpc{\Gamma}{a}{b}{\Lo{A}}{\Theta_1}{\Delta_1}}{\Sto{b}{\hat{B}} \in \Gamma & \So{b} \leq \Lo{A} &\JCl{\Gamma}{\Theta_1}{\Delta_1}}$$
    By the IH, $\Theta_1$ either steps, is poised, or contains a $\Psi$ that is stuck.

    If $\Theta_1$ steps, then ${\Lambda; \Theta_1 \to \Lambda'; \Theta_1'}$ for some $\Lambda'$ and $\Theta_1'$. Then\\
    ${\Lambda; \connect{a}{b}, \Theta_1 \to \Lambda'; \connect{a}{b}, \Theta_1'}$

    If $\Theta_1$ is poised, then $\connect{a}{b}, \Theta_1$ is poised because $\connect{-}{-}$ is not a proc term.

    Finally, if there is some $\Psi \in \Theta_1$ that is stuck, then course ${\Psi \in (\connect{a}{b}, \Theta_1)}$ is stuck.
  \end{case}
  \begin{case}
    $$\infer[\Theta 3]{\JCpp{\Gamma}{c}{P}{\Lo{C}}{\Theta_1}{\Delta_1}}
    {\Sto{c}{\hat{C}} \in \Gamma & \dsyncl{C'}{C}{\hat{C}} & \JTl{\Gamma}{\Delta_c}{P}{c}{C'} & \JCl{\Gamma}{\Theta_1}{\Delta_c, \Delta_1}}$$
    By the IH, $\Theta_1$ either steps, is poised, or contains a $\Psi$ that is stuck. We first cover two of the cases:

    If $\Theta_1$ steps, then ${\Lambda; \Theta_1 \to \Lambda'; \Theta_1'}$ for some $\Lambda'$ and $\Theta_1'$. Then\\
    ${\Lambda; \procl{c}{P}, \Theta_1 \to \Lambda'; \procl{c}{P}, \Theta_1'}$.

    If there is some $\Psi \in \Theta_1$ that is stuck, then of course the same ${\Psi \in (\procl{c}{P}, \Theta_1)}$ is stuck.

    For the final case, we will assume that $\Theta_1$ is poised and proceed by case analysis on the derivation of
    $\JTl{\Gamma}{\Delta_c}{P}{c}{C'}$.  Unlike in the first part, we make the step between identifying the appropriate typing rule and
    inferring the form of $P$ explicit because some of the cases are more complicated. In the typing judgment, we replace instantiated
    channel variables in the context such as $x$ by actual channel names since they must already exist in the configuration.
    \begin{subcase}
      The form of $P$ inferred from all linear right rules $(1R, {\otimes} R, {\otimes} \So{R}, {\multimap} R, {\oplus} R,
      {\&} R,\\ \Lupll R,$ and $\Ldownll R)$ directly coincide with the definition of poised. For example, $1R$ implies that $P = \close{a}$, which
      is poised, and so on. Since $\Theta_1$ is poised, $\procl{a}{P}, \Theta_1$ is poised.
    \end{subcase}
    \begin{subcase}
      $$\infer[{\otimes} L]{\JTl{\Gamma}{\Delta_c', \Lto{b}{\tensor{A}{B}}}{\asgn{\Lo{y}}{\recv{b}}; P}{c}{C'}}{\JTl{\Gamma}{\Delta_c', \Lt{b}{B}, \Lt{y}{A}}{P}{c}{C'}}$$
      \noindent where $\Delta_c = \Delta_c', \Lto{b}{\tensor{A}{B}}$. Then $\Theta_1 = \Theta_2, \procl{b}{-}, \Theta_3$ for some $\Theta_2$
      and $\Theta_3$ (we know $\Lo{b}$ is not provided by a connect term since connect terms offer channels of type $\upll{D}$). Since
      $\procl{b}{-}$ is poised and must offer a channel of type $\tensor{A}{B}$, it must be of form $\procl{b}{\sendl{b}{a}; Q}$. Thus, by
      \ref{dyn:tensor},
      \begin{gather*}
      \Lambda; \substack{\procl{c}{\asgn{\Lo{y}}{\recv{b}}; P}, \Theta_2, \\ \procl{b}{\sendl{b}{a}; Q}, \Theta_3} \to \Lambda;
      \substack{\procl{c}{[\Lsub{a}{y}]P}, \Theta_2,\\ \procl{b}{Q}, \Theta_3}
      \end{gather*}
      \noindent All the remaining linear left rules except $\Lupll L$ and $\Lupll R$ $(1L, {\multimap}L, {\multimap}\So{L}, {\oplus} L, {\&} L)$ follow a similar pattern.
    \end{subcase}
    \begin{subcase}
      $$\infer[\Lupls L]{\JTl{\Gamma, \Sto{a}{\hat{A}}}{\Delta_c}{\asgn{\Lo{x}}{\acqs{a}}; P}{c}{C'}}{\hat{A} \leq \upls{A} & \JTl{\Gamma, \Sto{a}{\hat{A}}}{\Delta, \Lt{x}{A}}{P}{c}{C'}}$$
      \noindent Since $\Lambda$ is poised, either ${\Lambda = \unavail{a}, \Lambda_1}$ or ${\Lambda = \procs{a}{\asgn{\Lo{x}}{\accs{a}}; Q},
      \Lambda_1}$ for some $\Lambda_1$. In the first case, $\procl{c}{\asgn{\Lo{a}}{\acqs{a}}; P}$ is stuck, so we are done. In the second case, by \ref{dyn:upls}, we have
      \begin{align*}
        &\procs{a}{\asgn{\Lo{x}}{\accs{a}}; Q}, \Lambda_1; \procl{c}{\asgn{\Lo{a}}{\acqs{a}}; P}, \Theta_1 \\
        \to \quad &\unavail{a}, \Lambda_1 ; \procl{c}{[\Lsub{a}{x}]P}, \procl{a}{[\Lsub{a}{x}]Q}, \Theta_1
      \end{align*}
    \end{subcase}
    \begin{subcase}
      $$\infer[\Ldownsl L]{\JTl{\Gamma}{\Delta_c', \Lto{a}{\downsl{A}}}{\asgn{\So{x}}{\rels{a}}; P}{c}{C'}}{\JTl{\Gamma, \St{x}{A}}{\Delta_c'}{P}{c}{C'}}$$
      \noindent where $\Delta_c = \Delta_c', \Lto{a}{\downsl{A}}$. Then $\Theta_1 = \Theta_2, \procl{a}{-}, \Theta_3$ for some $\Theta_2$ and $\Theta_3$. Since there is a $\procl{a}{-}$ in the linear configuration, by well-formedness condition, there must be a corresponding $\unavail{a} \in \Lambda$, so $\Lambda = \unavail{a}, \Lambda_1$. Furthermore, since $\Theta_1$ is poised, the proc term must be of form $\procl{a}{\asgn{\So{x}}{\dets{a}}; Q}$. By \ref{dyn:downsl}, we have
      \begin{align*}
                          &\unavail{a}, \Lambda_1; \procl{c}{\asgn{\So{x}}{\rels{a}}; P}, \Theta_2, \procl{a}{\asgn{\So{x}}{\dets{a}}; Q}, \Theta_3 \\
        \to \quad &\procs{a}{[\Ssub{a}{x}]Q}, \Lambda_1; \procl{c}{[\Lsub{a}{x}]P}, \Theta_2, \Theta_3
      \end{align*}
    \end{subcase}
    \begin{subcase}
      $$\infer[\Lupll L]{\JTl{\Gamma}{\Delta_c', \Lto{a}{\upll{A}}}{\asgn{\Lo{x}}{\acql{a}}; P}{c}{C'}}{\JTl{\Gamma}{\Delta_c', \Lt{x}{A}}{P}{c}{C'}}$$
      \noindent where $\Delta_c = \Delta_c', \Lto{a}{\upll{A}}$. Then $\Theta_1 = \Theta_2, \Psi_a, \Theta_3$ where $\Psi_a$ is either of
      form $\connect{a}{b}$ for some $\So{b}$ or $\procl{a}{-}$. In the latter case, we appeal to the term being poised and the proof
      proceeds like the other left rules. In the former case, there must be a term in $\Lambda$ that provides $\So{b}$. Since $\Lambda$ is
      poised, either ${\Lambda = \unavail{b}, \Lambda_1}$ or ${\Lambda = \procs{b}{\asgn{\Lo{x}}{\accs{b}}; Q}, \Lambda_1}$. In the former
      case, we can conclude that $\procl{c}{-}$ is stuck, so we are done. In the latter case, by \ref{dyn:upls2}, we have
      \begin{align*}
          &\procs{b}{\asgn{\Lo{x}}{\accs{b}}; Q}, \Lambda_1; \procl{c}{\asgn{\Lo{x}}{\acql{a}}; P}, \Theta_2, \connect{a}{b}, \Theta_3 \\
        \to \quad &\unavail{b}, \Lambda_1; \procl{c}{[\Lsub{b}{x}]P}, \procl{b}{[\Lsub{b}{x}]Q}, \Theta_2, \Theta_3
      \end{align*}
    \end{subcase}
    \begin{subcase}
      $$\infer[\Ldownll L]{\JTl{\Gamma}{\Delta_c', \Lto{a}{\downll{A}}}{\asgn{\Lo{x}}{\rell{a}}; P}{c}{C'}}{\JTl{\Gamma}{\Delta_c', \Lt{x}{A}}{P}{c}{C'}}$$
      \noindent where $\Delta_c = \Delta_c', \Lto{a}{\downll{A}}$. Then $\Theta_1 = \Theta_2, \procl{a}{-}, \Theta_3$. Since $\Theta_1$ is poised, there are two possible forms of $\procl{a}{-}$. If we have $\procl{a}{\asgn{\Lo{x}}{\detl{a}}; Q}$, then we appeal to the term being poised like the other left rules. If we instead have $\procl{a}{\asgn{\So{x}}{\dets{a}}; Q}$, then we first identify that $\Lambda = \unavail{a}, \Lambda_1$ for some $\Lambda_1$ by the well-formedness condition. By \ref{dyn:downsl2}, we have
      \begin{align*}
                &\unavail{a}, \Lambda_1; \procl{c}{\asgn{\Lo{x}}{\rell{a}}; P}, \Theta_2, \procl{a}{\asgn{\So{x}}{\dets{a}}; Q}, \Theta_3 \\
        \to \quad &\procs{a}{[\Ssub{a}{x}]Q}, \Lambda_1; \procl{c}{[\Lsub{b}{x}]P}, \connect{b}{a}, \Theta_2, \Theta_3 \fresh{b}
      \end{align*}
    \end{subcase}
  \end{case}
\end{proof}

\end{document}